\documentclass[twocolumn,secnumarabic,amsmath,amsfonts,amssymb,nobibnotes,aps,prb,superscriptaddress,floatfix,footinbib]{revtex4-2}

\usepackage{amsthm}
\usepackage{array}
\usepackage{booktabs}
\usepackage{hyperref}
\usepackage{enumerate}
\usepackage{graphicx}
\usepackage{mathrsfs}
\usepackage{mathtools}
\usepackage{newtxtext}
\usepackage{newtxmath}
\usepackage{subfigure}
\usepackage{tikz}
\usepackage{ulem} 
\usepackage{xcolor}
\usetikzlibrary{shapes.geometric, arrows}
\usetikzlibrary{positioning}
\usetikzlibrary{calc}
\hypersetup{colorlinks,breaklinks}
\hypersetup{
	colorlinks=true,
	linkcolor=[rgb]{0,0,1},
	filecolor=blue,      
	urlcolor=[rgb]{0,0,1},
	citecolor=[rgb]{0,0,1},
}
\renewenvironment{proof}{{\noindent\it Proof.}}{\hfill $\blacksquare$\par}
\tikzstyle{startstop} = [rectangle, rounded corners, minimum width = 0.2cm, minimum height=0.2cm,text centered, draw = black, fill = red!40]
\tikzstyle{io} = [trapezium, rounded corners, trapezium left angle=70, trapezium right angle=110, minimum width=1cm, minimum height=0.5cm, text centered, draw=black, fill = blue!40]
\tikzstyle{process} = [rectangle, rounded corners, minimum width=0.5cm, minimum height=0.4cm, text centered, draw=black, fill = yellow!50]
\tikzstyle{decision} = [diamond, rounded corners, aspect=3.7,text centered, draw=black, align = center, minimum width = 0.5cm, minimum height = 0.4cm,fill = green!30]
\tikzstyle{arrow} = [->,>=stealth,rounded corners]

\newcommand{\PreserveBackslash}[1]{\let\temp=\\#1\let\\=\temp}
\newcolumntype{C}[1]{>{\PreserveBackslash\centering}p{#1}}
\newcolumntype{R}[1]{>{\PreserveBackslash\raggedleft}p{#1}}
\newcolumntype{L}[1]{>{\PreserveBackslash\raggedright}p{#1}}

\newcommand{\ab}{\textit{ab initio}}
\newcommand{\st}{\mathrm{s.}~\mathrm{t.}}
\newcommand{\lrbrace}[1]{\left\{#1\right\}}

\newcommand{\lrsquare}[1]{\left[#1\right]}

\newcommand{\white}[1]{\textbf{\color{white}{#1}}}
\newcommand{\T}{\top}
\newcommand{\trace}{\mathrm{Tr}}
\newcommand{\dev}{\textup{\text{dev}}}

\newcommand{\atom}{\textup{\text{atom}}}
\newcommand{\inter}{\textup{\text{inter}}}
\newcommand{\latt}{\textup{\text{latt}}}

\newcommand{\trial}{\textup{\text{trial}}}
\newcommand{\Dlak}{\tilde D_{\latt}^{(k)}}

\newcommand{\aaFlak}{\tilde F_{\latt}^{(k)}}
\newcommand{\Ainterkl}{A_{\inter}^{(k,\ell)}}

\newcommand{\alatkl}{\alpha_{\atom}^{(k-1,\ell)}}
\newcommand{\allakl}{\alpha_{\latt}^{(k-1,\ell)}}
\newcommand{\Atkl}{A_{\trial}^{(k,\ell)}}
\newcommand{\Rtkl}{R_{\trial}^{(k,\ell)}}
\newcommand{\eps}{\varepsilon}
\newcommand{\F}{\mathrm{F}}
\newcommand{\calF}{\mathcal{F}}
\newcommand{\calT}{\mathcal{T}}

\newcommand{\mb}[1]{\mathbf{#1}}
\newcommand{\aaa}{\mb{a}}
\newcommand{\rr}{\mb{r}}
\newcommand{\norm}[1]{\left\Vert#1\right\Vert}
\newcommand{\snorm}[1]{\Vert#1\Vert}
\newcommand{\abs}[1]{\left|#1\right|}
\newcommand{\inner}[1]{\left\langle#1\right\rangle}

\newcommand{\R}{\mathbb{R}}
\newcommand{\N}{\mathbb{N}}

\newcommand{\abb}{\text{ABB}}
\newcommand{\bb}{\text{BB}}

\newcommand{\calI}{\mathcal{I}}
\newcommand{\ax}{\text{ax}}

\newcommand{\y}{\mb{y}}
\newcommand{\magg}{\text{mag}}
\newcommand{\dd}{\,\mathrm{d}}

\DeclareMathOperator*{\argmin}{arg\,min}

\newcommand{\angstrom}{\text{\normalfont\AA}}

\newtheorem{thm}{Theorem}

\newtheorem{lem}{Lemma}

\newtheorem{assume}{Assumption}

\begin{document}
	
	\title{Projected gradient descent algorithm for \ab~crystal structure relaxation\\under a fixed unit cell volume}
	
	\author{Yukuan Hu}%
	\affiliation{State Key Laboratory of Scientific and Engineering Computing, Institute of Computational Mathematics and Scientific/Engineering Computing, Academy of Mathematics and Systems Science, Chinese Academy of Sciences, Beijing 100190, China}%
	\affiliation{School of Mathematical Sciences, University of Chinese Academy of Sciences, Beijing 100049, China}%
	\author{Junlei Yin}%
	\affiliation{State Key Laboratory of Solidification Processing, Northwestern Polytechnical University, Xi'an 710072, Shaanxi, China}%
	\affiliation{Innovation Center, NPU Chongqing, Chongqing 401135, China}
	\author{Xingyu Gao}%
	\email[Corresponding author: ]{gao\_xingyu@iapcm.ac.cn}
	\affiliation{Laboratory of Computational Physics, Institute of Applied Physics and Computational Mathematics, Beijing 100088, China}%
	\author{Xin Liu}%
	\email[Corresponding author: ]{liuxin@lsec.cc.ac.cn}
	\affiliation{State Key Laboratory of Scientific and Engineering Computing, Institute of Computational Mathematics and Scientific/Engineering Computing, Academy of Mathematics and Systems Science, Chinese Academy of Sciences, Beijing 100190, China}%
	\affiliation{School of Mathematical Sciences, University of Chinese Academy of Sciences, Beijing 100049, China}%
	\author{Haifeng Song}%
	\email[Corresponding author: ]{song\_haifeng@iapcm.ac.cn}
	\affiliation{Laboratory of Computational Physics, Institute of Applied Physics and Computational Mathematics, Beijing 100088, China}%
	
	\begin{abstract}
		This paper is concerned with \ab~crystal structure relaxation under a fixed unit cell volume, which is a step in calculating the static equations of state and forms the basis of thermodynamic property calculations for materials. The task can be formulated as an energy minimization with a determinant constraint. Widely used line minimization-based methods (e.g., conjugate gradient method) lack both efficiency and convergence guarantees due to the nonconvex nature of the feasible region as well as the significant differences in the curvatures of the potential energy surface with respect to atomic and lattice components. To this end, we propose a projected gradient descent algorithm named PANBB. It is equipped with (i) search direction projections onto the tangent spaces of the nonconvex feasible region for lattice vectors, (ii) distinct curvature-aware initial trial step sizes for atomic and lattice updates, and (iii) a nonrestrictive line minimization criterion as the stopping rule for the inner loop. It can be proved that PANBB favors theoretical convergence to equilibrium states. Across a benchmark set containing 223 structures from various categories, PANBB achieves average speedup factors of approximately 1.41 and 1.45 over the conjugate gradient method and direct inversion in the iterative subspace implemented in off-the-shelf simulation software, respectively. Moreover, it normally converges on all the systems, manifesting its unparalleled robustness. As an application, we calculate the static equations of state for the high-entropy alloy AlCoCrFeNi, which remains elusive owing to 160 atoms representing both chemical and magnetic disorder and the strong local lattice distortion. The results are consistent with the previous calculations and are further validated by experimental thermodynamic data.
	\end{abstract}
	
	\maketitle
	
	\section{Introduction} 
	
	Crystal structure relaxation seeks equilibrium states on the potential energy surface (PES). It underpins various applications such as crystal structure predictions \cite{oganov2006crystal,wang2010crystal,chen2017sgo,wang2022crystal} and high-throughput calculations in material design \cite{goedecker2005global,vilhelmsen2012systematic,hu2013pressure,li2013global,zhang2013genetic}. Relaxing crystal structures under a fixed unit cell volume is a step in calculating the static equations of state (EOS) for materials \cite{vinet1989universal,jackson2015crystal} and also presents the volume-dependent configurations for thermal models such as the quasiharmonic approximation \cite{born1988dynamical, carrier2007first}, classical mean-field potential \cite{wang2000calculated, wang2000classical, wang2000mean, li2001thermodynamic, song2007modified}, and electronic excitation based on Mermin statistics \cite{wasserman1996, jarlborg1997}, etc. 
	
	\par Mathematically, relaxing the crystal structure under a fixed unit cell volume can be formulated as minimizing the potential energy with a determinant constraint, where a nonconvex feasible region arises. Developing convergent methods for general nonconvex constrained optimization problems remains challenging, since classical tools (e.g., projection) and theories (e.g., duality) become inapplicable. As far as we know, little attention has been paid to the determinant-constrained optimization.
	
	\par In practice, line minimization (LM)-based methods are usual choices for this purpose. A flowchart of the widely used LM-based methods is depicted in Fig. \ref{fig:existing methods flowchart}. Each iteration of these methods primarily consists of updating the search directions followed by an LM process to determine the step sizes. Among the popular representatives in this class are the conjugate gradient method (CG) \cite{shewchuk1994introduction} and direct inversion in the iterative subspace (DIIS) \cite{pulay1980convergence}. Specifically, CG searches along the conjugate gradient directions made up by previous directions and current forces, while DIIS employs historical information to construct the quasi-Newton directions. In most cases, CG exhibits steady convergence but can be unsatisfactory in efficiency. DIIS converges rapidly in the small neighborhoods of local minimizers but can easily diverge when starting far from equilibrium states.
	
	\par As shown in Fig. \ref{fig:existing methods flowchart}, the performance of the LM-based methods crucially depends on three key factors: search directions, initial trial step sizes, and LM criterion. However, the fixed-volume relaxation presents substantial difficulties on their designs. Firstly, due to the nonconvex nature of the feasible region, employing search directions for lattice vectors without tailored modifications can result in considerable variations in the configurations, {giving rise to} drastic fluctuations in the potential energies; for an illustration, please refer to Fig. \ref{fig:cg full test} later. Secondly, on account of the 
	significant differences in the curvatures of PES with respect to the atomic and lattice components, adopting shared initial trial step sizes, as in the widely used LM-based methods (see S2 in Fig. \ref{fig:existing methods flowchart}), can be detrimental to algorithm efficiency. This point is recognized through extensive numerical tests, where the relative displacements in the lattice vectors during relaxation are typically much smaller than those in the atomic positions. We refer readers to Fig. \ref{fig:block compare} later for an illustration. Thirdly, the monotone criterion utilized in the widely used LM-based methods mandates a strict reduction in the potential energy after each iteration, often making the entire procedure trapped in the LM process. According to incomplete statistics, the CG implemented in off-the-shelf simulation software spends approximately 60\% of its overhead on the trials rejected by the LM criterion, as supported by the Supplemental Material. It is worth noting that the authors of \cite{hu2022force} devise a nonmonotone LM criterion but only consider atomic degrees of freedom. 
	
	\begin{figure}[!t]
		\centering
		\begin{tikzpicture}[node distance=1.2cm,font=\footnotesize]
			\draw [dashed,rounded corners,fill = gray!30] (-3.4cm,-4.4cm) rectangle (4.65cm,-10.55cm);
			\node [startstop,minimum width = 1.5cm] (start) {Start};
			\node [io, below of = start, align = center,yshift=0.2cm] (input) {Input initial configuration $(R^{(0)},A^{(0)})$,\\unit cell volume $V$, stretching factor $s>0$,\\force tolerance $\eps>0$; set $k:=0$};
			\node [decision, below of = input, yshift=-0.55cm] (tolerance) {Force tolerance satisfied?};
			\node [startstop, right of = tolerance, xshift=2.2cm] (stop) {Stop};
			\node [process, below of = tolerance, yshift=-0.1cm, align=center] (updatedir) {{\bf S1.} Update \textbf{search directions} $D_{\atom}^{(k)}$ and $D_{\latt}^{(k)}$};
			\node [process, below of = updatedir,align = center,yshift=0.25cm] (compstep) {{\bf S2.} Compute an \textbf{initial trial step size} $\alpha^{(k-1,0)}$; set $\ell:=0$};
			\node [process, below of = compstep, align = center,yshift=-0.2cm] (updatefactorconfig) {{\bf S3.} Compute a new trial atomic configuration\\ $\Rtkl:=R^{(k)}+\alpha^{(k-1,\ell)}D_{\atom}^{(k)}$\\and a new intermediate lattice matrix\\$\Ainterkl:=A^{(k)}+\alpha^{(k-1,\ell)}D_{\latt}^{(k)}$};
			\node [process, below of = updatefactorconfig, yshift=-0.5cm, align=center] (scaling) {{\bf S4.} Perform the scaling operation\\ $\Atkl:=\sqrt[3]{\frac{V}{\det\big(\Ainterkl\big)}}\Ainterkl$};
			\node [process, right of = scaling, align = center, xshift = 2.15cm] (updatestep) {{\bf S6.} Update $\alpha^{(k-1,\ell)}$\\to $\alpha^{(k-1,\ell+1)}$;\\set $\ell:=\ell+1$};
			\node [decision, below of = scaling,yshift=-.55cm] (ls) {{\bf S5.} \textbf{LM criterion}\\accepts trial?};
			\node [process, below of = ls,yshift=-0.3cm,align=center] (updateconfig) {{\bf S7.} Update the step size $\alpha^{(k)}:=\alpha^{(k-1,\ell)}$ and configuration\\$(R^{(k+1)},A^{(k+1)}):=(\Rtkl,\Atkl)$; set $k:=k+1$};
			\draw [arrow] (start) -- (input);
			\draw [arrow] (input) -- node [process,right,rounded corners = 0pt,fill=blue!80,xshift=0.1cm] {\white{KS}} (tolerance);
			\draw [arrow] (tolerance) -- node [above] {Yes} (stop);
			\draw [arrow] (tolerance) -- node [right] {No} (updatedir);
			\draw [arrow] (updatedir) -- (compstep);
			\draw [arrow] (compstep) -- (updatefactorconfig);
			\draw [arrow] (updatefactorconfig) -- (scaling);
			\draw [arrow] (scaling) -- (ls);
			\draw [arrow] (ls) -- node [right] {Yes} (updateconfig);
			\draw [arrow] (updateconfig.west) -- ++ (-0.5cm,0) -- ++ (0,7.5cm) |- (tolerance.west); 
			\draw [arrow] (ls.east) -- ++ (0.5cm,0) -| (updatestep.south);
			\draw [arrow] (updatestep.north) -- ++ (0,.8cm) |- (updatefactorconfig.east);
			\node [right of = ls, xshift = 1.8cm, yshift = 0.15cm] {No};
			\node [process, below of = scaling, xshift = 0.4cm, yshift = 0.35cm, rounded corners = 0pt, fill = blue!80] {\white{KS}};
			\node at (3.05cm,-4.57cm) {\textbf{\textit{Line minimization process}}};
		\end{tikzpicture}
		\caption{A flowchart of the widely used LM-based methods. The bold text ``\textbf{KS}'' refers to solving Kohn-Sham (KS) equations to obtain new energy, atomic forces, and lattice stress whenever the configuration gets updated. The LM process is marked out by the shaded box. The subscript ``inter'' indicates the intermediate lattice matrix.}
		\label{fig:existing methods flowchart}
	\end{figure}
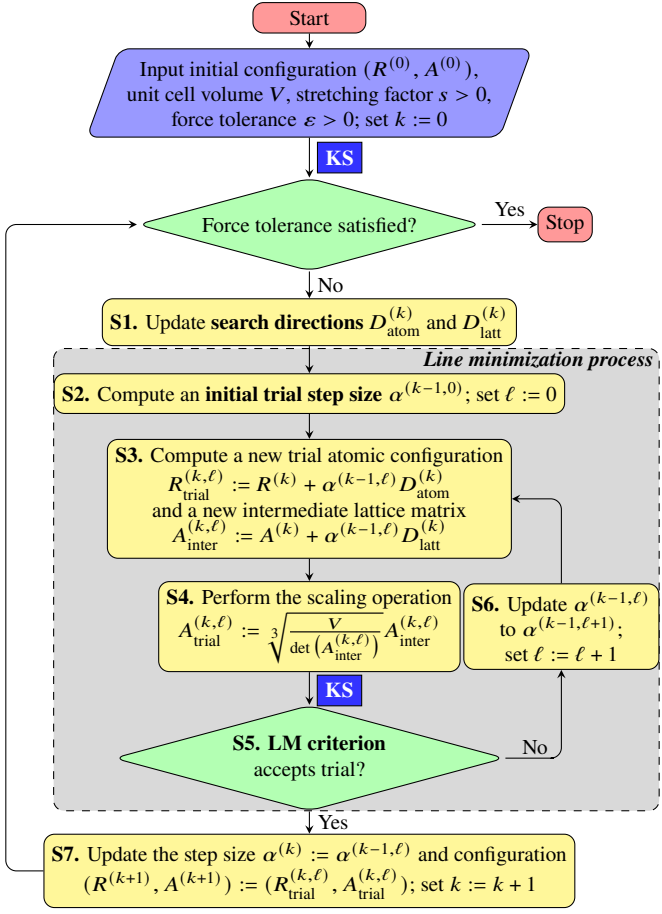 
	
	\par In this work, we propose a projected gradient descent algorithm, called PANBB, that addresses all the aforementioned difficulties. Specifically, PANBB (i) employs search direction projections onto the tangent spaces of the nonconvex feasible region for lattice vectors, (ii) adopts distinct curvature-aware step sizes for the atomic and lattice components, and (iii) incorporates a nonmonotone LM criterion that encompasses both atomic and lattice degress of freedom. The convergence of PANBB to equilibrium states has been theoretically established. 
	
	\par PANBB has been tested across a benchmark set containing 223 systems from various categories, with approximately 68.6\% of them are metallic systems, for which the EOS calculations are more demanding. PANBB exhibits average speedup factors of approximately 1.41 and 1.45 over the CG and DIIS implemented in off-the-shelf simulation software, respectively, in terms of running time. In contrast to the failure rates of approximately 4.9\% for CG and 25.1\% for DIIS, PANBB consistently converges to the equilibrium states across the broad, manifesting its unparalleled robustness. As an application of PANBB, we calculate the static EOS of the high-entropy alloy (HEA) AlCoCrFeNi, whose fixed-volume relaxation remains daunting owing to the strong local lattice distortion (LLD) \cite{song2017local}. The numerical results are consistent with the previous calculations \cite{wu2020structural} obtained by relaxing only atomic positions. We further apply the modified mean field potential (MMFP) approach \cite{song2007modified} to the thermal EOS, which is validated by the X-ray diffraction and diamond anvil cell experiments \cite{cheng2019pressure}.
	
	\par The notations are gathered as follows. We use $N\in\N$ for the number of atoms, and use $R:=[\rr_1,\ldots,\rr_N]\in\R^{3\times N}$ and $A:=[\aaa_1,\aaa_2,\aaa_3]\in\R^{3\times3}$ to denote the Cartesian atomic positions and lattice matrix, respectively. The reciprocal lattice matrix $A^{-\T}$ is abbreviated as $B\in\R^{3\times3}$. We denote by $E(R,A)\in\R$, $F_{\atom}(R,A)\in\R^{3\times N}$, $F_{\latt}(R,A)\in\R^{3\times3}$, $\Sigma(R,A)\in\R^{3\times3}$, and $\Sigma_{\dev}(R,A)\in\R^{3\times3}$, respectively, the potential energy, atomic forces, lattice force, full stress tensor, and deviatoric stress tensor evaluated at the configuration $(R,A)$; if no confusion arises, we omit specifying $(R,A)$. From \cite{atalla2013density,knuth2015strain}, one obtains the relation
	\begin{equation}
		F_{\latt}=V\Sigma B-F_{\atom}R^\T B.
		\label{eqn:lattice force}
	\end{equation}
	When describing algorithms, we use the first and second superscripts in brackets to indicate respectively the outer and inner iteration numbers (e.g., $E^{(k)}$ and $\alpha^{(k-1,\ell)}$). We denote by ``$\mb{I}$'' the identity mapping from $\R^{3\times N}$ to $\R^{3\times N}$. The notation ``$\inner{\cdot,\cdot}$'' refers to the standard inner product between two matrices, defined as $\inner{M,N}:=\trace(M^\T N)$, whereas ``$\snorm{\cdot}_\F$'' yields the matrix Frobenius norm as $\snorm{M}_\F:=\sqrt{\inner{M,M}}$. 
	
	\section{Algorithmic Developments}
	
	\par In the sequel, we expound on our algorithmic developments from three aspects: search directions, initial trial step sizes, and LM algorithm and criterion.
	
	\subsection{Search Directions}
	
	\par For the atomic part, we simply take the steepest descent directions, i.e., 
	\begin{equation*}
		D_{\atom}^{(k)}:=F_{\atom}^{(k)}~\text{for all}~k\ge0.
	\end{equation*}
	As for the lattice part, we first note that the steepest descent direction $F_{\latt}^{(k)}$ (see Eq. \eqref{eqn:lattice force}) cannot be used directly due to the volume constraint; otherwise, there can be unexpected large variations in the lattice vectors, resulting in drastic fluctuations in the potential energies.
	\begin{figure*}[!t]
		\centering
		\includegraphics[width=.45\linewidth]{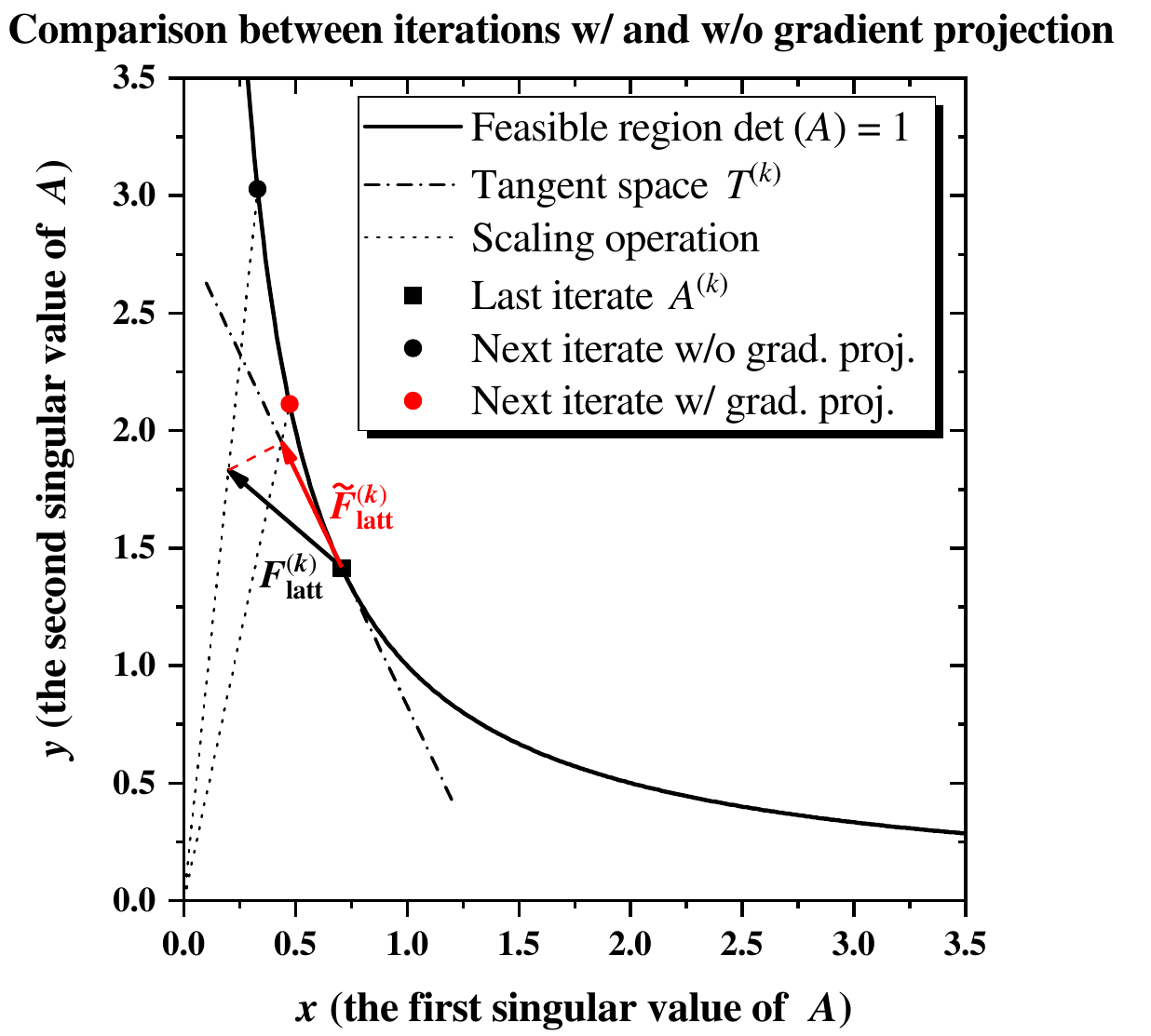}\quad\quad
		\includegraphics[width=.445\linewidth]{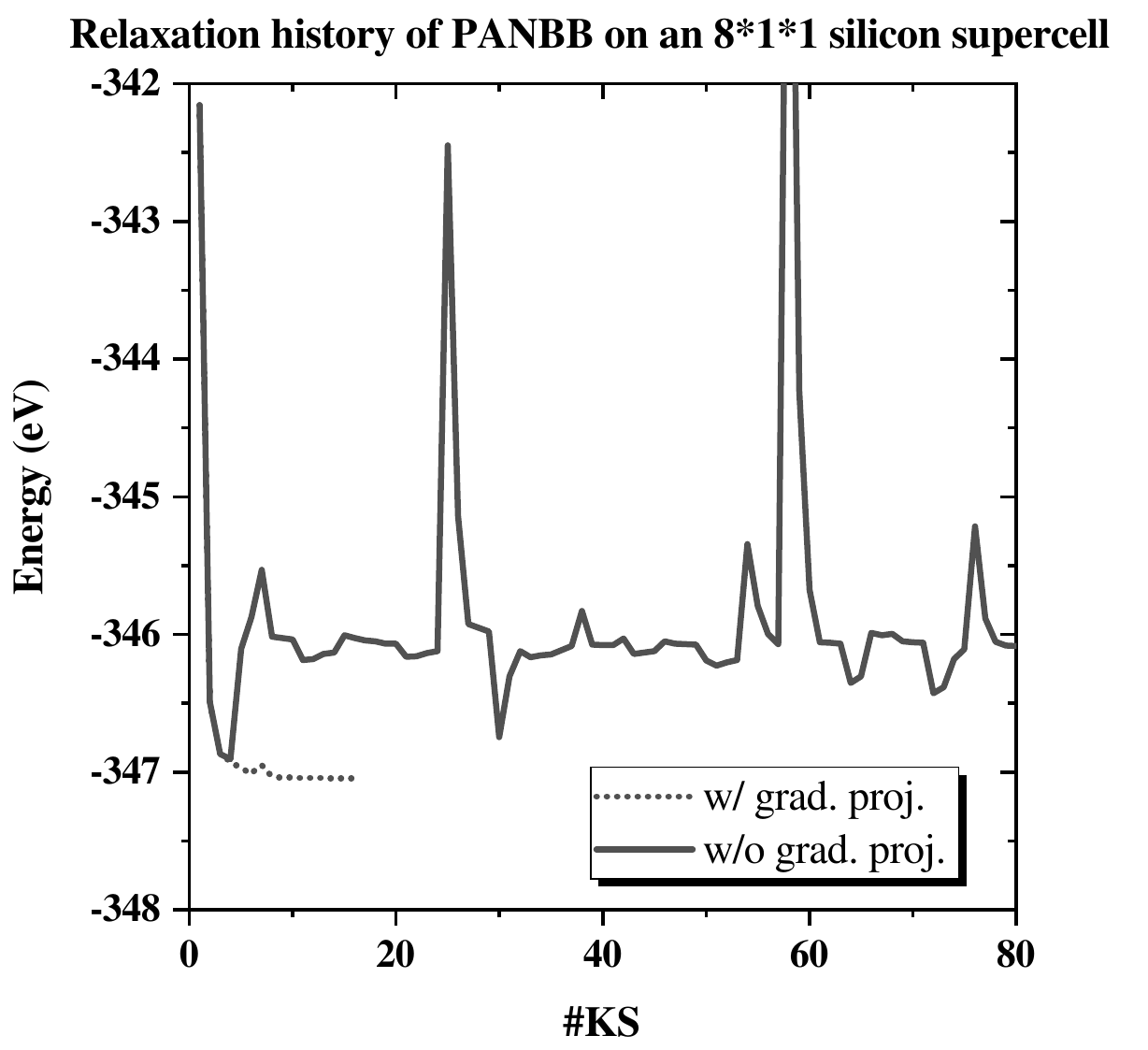}
		\caption{Left: a comparison between iterations with and without the gradient projection \eqref{eqn:lattice force projection} in two-dimensional context. The $x$- and $y$-axes refer to the first and second singular values of lattice matrix $A$, respectively. The black solid line stands for the feasible region $\det(A)=1$, the black square denotes the last iterate $A^{(k)}$, and the black dashdotted line refers to the tangent space $\calT^{(k)}$ of feasible region at $A^{(k)}$. The black and red arrows represent respectively the lattice force $F_{\latt}^{(k)}$ and its projection onto $\calT^{(k)}$, namely, $\tilde{F}_{\latt}^{(k)}$. The black and red circles show the next iterates without and with the gradient projection \eqref{eqn:lattice force projection}, respectively, after performing the scaling operation indicated by the black dotted lines. Right: the relaxation history of PANBB on an $8\times1\times1$ silicon supercell (64 atoms), where ``\#KS'' denotes the number of solving the KS equations, ``w/ grad. proj.'' and ``w/o grad. proj.'' stand for the results using $\tilde F_{\latt}^{(k)}$ and $F_{\latt}^{(k)}$, respectively.}
		\label{fig:cg full test}
	\end{figure*}
	For an illustration in two-dimensional context, see the left panel of Fig. \ref{fig:cg full test} where one can observe a large distance between the last iterate $A^{(k)}$ (denoted by the black square) and the next iterate using $F_{\latt}^{(k)}$ (denoted by the black circle). 
	
	\par Instead, we search along the tangent lines of the feasible region at the current configuration. In mathematics, this amounts to projecting the lattice force $F_{\latt}^{(k)}$ onto the tangent space $\calT^{(k)}\subseteq\R^{3\times3}$ \cite{nocedal2006numerical} of $\calF:=\{A\in\R^{3\times3}:\det(A)=V\}$ at $A^{(k)}$, where
	\begin{equation}
		\calT^{(k)}:=\lrbrace{D\in\R^{3\times3}:\inner{B^{(k)},D}=0}.
		\label{eqn:local linear approx}
	\end{equation}
	The projection favors a closed-form expression:
	\begin{equation}
		\aaFlak:=F_{\latt}^{(k)}-\frac{\inner{B^{(k)},F_{\latt}^{(k)}}}{\snorm{B^{(k)}}_\F^2}B^{(k)}.
		\label{eqn:lattice force projection}
	\end{equation}
	We then take $D_{\latt}^{(k)}:=\aaFlak$ for $k\ge0$. A comparison between the iterations with and without the gradient projection \eqref{eqn:lattice force projection} can be found in the left panel of Fig. \ref{fig:cg full test}. After projecting $F_{\latt}^{(k)}$ (denoted by the black arrow) onto $\calT^{(k)}$, we obtain $\tilde{F}_{\latt}^{(k)}$ (denoted by the red arrow) and then an iterate closer to $A^{(k)}$ (denoted by the red circle).
	
	\par It can be shown theoretically that the gradient projection \eqref{eqn:lattice force projection} eliminates the potential first-order energy increment induced by the scaling operation (S4 in Fig. \ref{fig:existing methods flowchart}). We shall point out that the strategy does not necessarily apply to general nonconvex constrained optimization problems. The desired theoretical results are achieved by fully exploiting the algebraic characterization of the tangent space \eqref{eqn:local linear approx}. For technical details, please refer to the proof of Lemma \ref{lem:successfully generation} in Appendix \ref{appsec:convergence analysis}. We also demonstrate the merit of the gradient projection \eqref{eqn:lattice force projection} through numerical comparisons; see the right panel of Fig. \ref{fig:cg full test}. 
	
	\subsection{Initial Trial Step Sizes}
	
	\par Efficient initial trial step sizes are essential for improving the overall performance. Moreover, our extensive numerical tests reveal that the relative displacements in the lattice vectors are typically much smaller than those in the atomic positions during relaxation, particularly for large systems. This observation underscores the substantial differences in the local curvatures of PES with respect to the atomic positions and lattice vectors, or in other words, ill conditions underlying the relaxation tasks. In light of this, we shall treat the atomic and lattice components separately and leverage their respective curvature information.
	
	\begin{figure*}[!t]
		\centering
		\includegraphics[width=.45\linewidth]{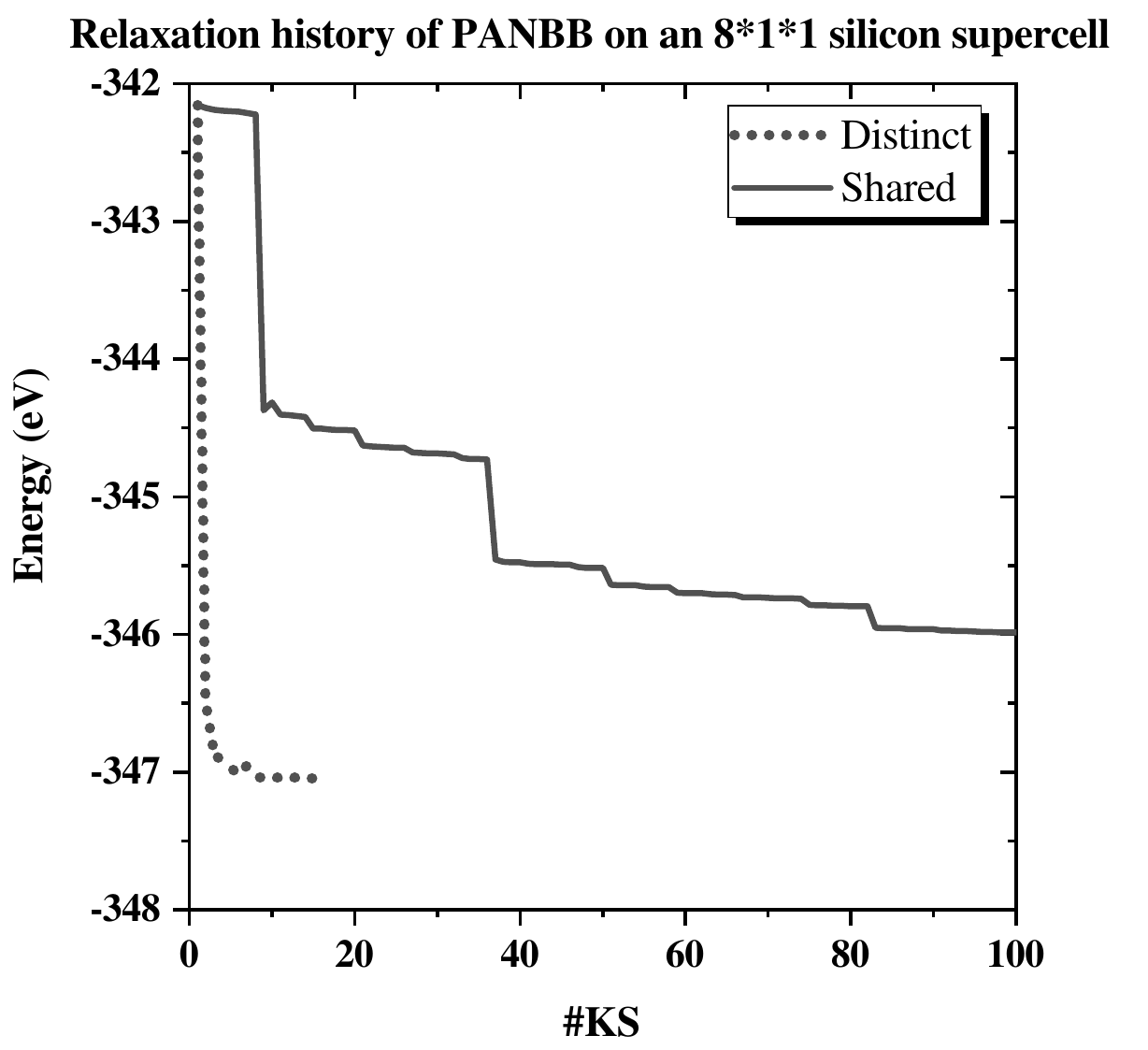}\quad\quad
		\includegraphics[width=.375\linewidth]{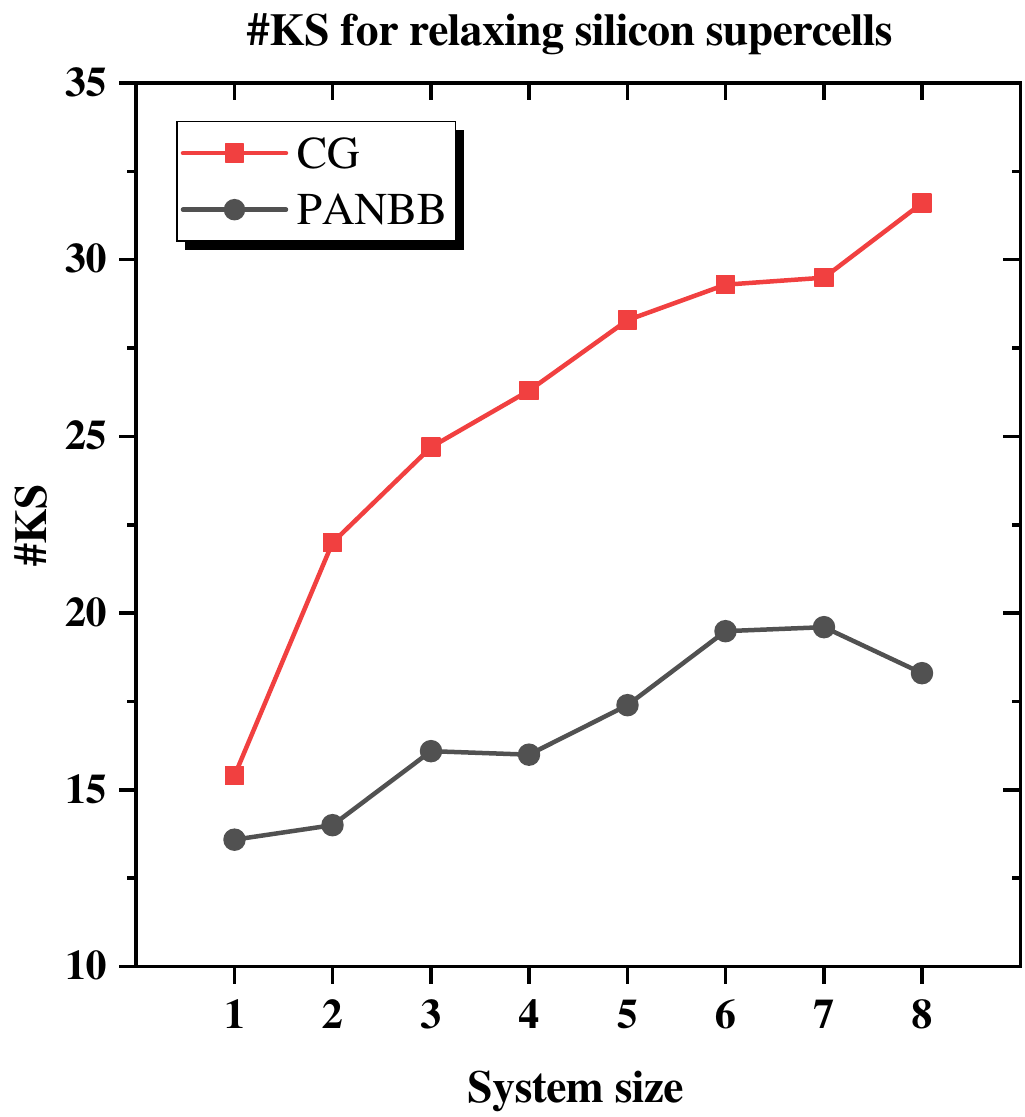}
		\caption{The relaxation on silicon supercells ($8\sim64$ atoms). Left: the relaxation history of PANBB on an $8\times1\times1$ silicon supercell. ``Distinct'' and ``Shared'' stand for the results of PANBB using distinct and shared step sizes for atomic and lattice components, respectively. Right: the average \#KS by CG and PANBB for relaxing silicon supercells of different sizes. Each system size is associated with 10 initial configurations, which are generated by random perturbations from equilibrium states; the deviations in the atomic positions are less than $0.1$ \angstrom~and the deviations in the lattice vectors are less than $0.02$ \angstrom.}
		\label{fig:block compare}
	\end{figure*}
	
	\par Motivated by the above discussions, we introduce curvature-aware step sizes for the atomic and lattice components, respectively. For the atomic part, we adopt the alternating Barzilai-Borwein (ABB) step sizes \cite{barzilai1988two,dai2005projected,hu2022force}, defined as
	\begin{equation}
		\hspace{-1mm}\alpha_{\atom}^{(k-1,0)}=\alpha_{\atom,\abb}^{(k)}:=\left\{\begin{array}{ll}
			\alpha_{\atom,\bb1}^{(k)}, & \text{if}~\mathrm{mod}(k,2)=0;\\
			\alpha_{\atom,\bb2}^{(k)}, & \text{if}~\mathrm{mod}(k,2)=1
		\end{array}\right.
		\label{eqn:ABB atom}
	\end{equation}
	for $k\ge1$, where 
	$$\alpha_{\atom,\bb1}^{(k)}:=\frac{\snorm{S_{\atom}^{(k-1)}}_\F^2}{\inner{S_{\atom}^{(k-1)},Y_{\atom}^{(k-1)}}},~\alpha_{\atom,\bb2}^{(k)}:=\frac{\inner{S_{\atom}^{(k-1)},Y_{\atom}^{(k-1)}}}{ \snorm{Y_{\atom}^{(k-1)}}_\F^2},$$
	$S_{\atom}^{(k-1)}:=R^{(k)}-R^{(k-1)}$, and $Y_{\atom}^{(k-1)}:=F_{\atom}^{(k-1)}-F_{\atom}^{(k)}$. It is not difficult to verify that
	$$\begin{aligned}
		\alpha_{\atom,\bb1}^{(k)}&=\argmin_{\alpha}\norm{\alpha^{-1}S_{\atom}^{(k-1)}-Y_{\atom}^{(k-1)}}_\F,\\
		\alpha_{\atom,\bb2}^{(k)}&=\argmin_{\alpha}\norm{\alpha Y_{\atom}^{(k-1)}-S_{\atom}^{(k-1)}}_\F.
	\end{aligned}$$
	In other words, $\alpha_{\atom,\bb1}^{(k)}\mb{I}$ and $\alpha_{\atom,\bb2}^{(k)}\mb{I}$ approximate the inverse atomic Hessian of potential energy around $R^{(k)}$. The gradient descent with the above BB step sizes thus resembles the Newton iteration for solving the nonlinear equations $F_{\atom}=0$.
	
	\par Now let us derive the ABB step sizes for the lattice part. Due to the volume constraint, the equilibrium condition on the lattice becomes $\Sigma_{\dev}=0$, namely, the shear stress vanishes and the normal stress cancels out the external isotropic pressure. This is equivalent to $\tilde F_{\latt}=0$ in view of Eqs. \eqref{eqn:lattice force} and \eqref{eqn:lattice force projection}, provided that the atomic forces vanish. Recalling the Newton iteration for solving the nonlinear equations $\tilde F_{\latt}=0$, we obtain the following ABB step sizes for the lattice vectors
	\begin{equation}
		\alpha_{\latt}^{(k-1,0)}=\alpha_{\latt,\abb}^{(k)}:=\left\{\begin{array}{ll}
			\alpha_{\latt,\bb1}^{(k)}, & \text{if}~\mathrm{mod}(k,2)=0;\\
			\alpha_{\latt,\bb2}^{(k)}, & \text{if}~\mathrm{mod}(k,2)=1
		\end{array}\right.
		\label{eqn:ABB lattice}
	\end{equation}
	for $k\ge1$, where
	$$\alpha_{\latt,\bb1}^{(k)}:=\frac{\snorm{S_{\latt}^{(k-1)}}_\F^2}{\inner{S_{\latt}^{(k-1)},Y_{\latt}^{(k-1)}}},~\alpha_{\latt,\bb2}^{(k)}:=\frac{\inner{S_{\latt}^{(k-1)},Y_{\latt}^{(k-1)}}}{\snorm{Y_{\latt}^{(k-1)}}_\F^2},$$
	$S_{\latt}^{(k-1)}:=A^{(k)}-A^{(k-1)}$, and $Y_{\latt}^{(k-1)}:=\tilde F_{\latt}^{(k-1)}-\tilde F_{\latt}^{(k)}$. With these distinct step sizes, we update the trials as
	\begin{equation}
		\begin{aligned}
			R_{\trial}^{(k,\ell)}:&=R^{(k)}+\alpha_{\atom}^{(k-1,\ell)}D_{\atom}^{(k)},\\
			A_{\inter}^{(k,\ell)}:&=A^{(k)}+\alpha_{\latt}^{(k-1,\ell)}D_{\latt}^{(k)},\\
			A_{\trial}^{(k,\ell)}:&=\sqrt[3]{\frac{V}{\det\big(A_{\inter}^{(k,\ell)}\big)}}A_{\inter}^{(k,\ell)},
		\end{aligned}
		\quad\ell=0,1,\ldots.
		\label{eqn:block update}
	\end{equation}
	
	\par As shown in Fig. \ref{fig:block compare}, we demonstrate the benefits brought by the distinct ABB step sizes \eqref{eqn:ABB atom} and \eqref{eqn:ABB lattice} through the relaxation of silicon supercells. From the left panel, one can observe the considerable acceleration of the distinct version over the shared one. In the latter case, the ABB step sizes are shared and computed by the quantities from both the atomic and lattice components. In the right panel, we test our new 
	algorithm equipped with the distinct step sizes on the silicon supercells of different sizes ($1\times1\times1\sim8\times1\times1$). Compared with that of the CG built in off-the-shelf simulation software, the number of solving the KS equations of the new algorithm increases slowly as the system size grows, which indicates a preconditioning-like effect of the our strategy. 
	
	\subsection{LM Algorithm and Criterion}
	
	\par With the search directions in place, the LM algorithm starts from the initial step sizes and attempts to locate an appropriate trial configuration that meets certain LM criterion. Capturing the local PES information, however, the ABB step sizes can occasionally incur small increments on energies \cite{hu2022force}. Extensive numerical simulations suggest that keeping the ABB step sizes intact is more beneficial \cite{fletcher2005barzilai,hu2022force}. The monotone LM criterion (S5 in Fig. \ref{fig:existing methods flowchart}) adopted in the widely used LM-based methods will lead to unnecessary computational expenditure and undermine the merits of the ABB step sizes. Notably, the authors of \cite{hu2022force} propose a nonmonotone criterion which encompasses only atomic degress of freedom.
	
	\par In the following, a nonrestrictive nonmonotone LM criterion is presented for both atomic and lattice degrees of freedom and the simple backtracking is taken as the LM algorithm. Specifically, the nonmonotone criterion accepts the trial configuration if
	\begin{equation}
		E_{\trial}^{(k,\ell)}\le \bar E^{(k)}-\eta\big(\alpha_{\atom}^{(k-1,\ell)}\snorm{F_{\atom}^{(k)}}_\F^2+\alpha_{\latt}^{(k-1,\ell)}\snorm{\aaFlak}_\F^2\big),
		\label{eqn:nonmonotone termination criterion}
	\end{equation}
	where $E_{\trial}^{(k,\ell)}:=E(\Rtkl,\Atkl)$, $\eta\in(0,1)$ is a constant, and $\{\bar E^{(k)}\}$ is a surrogate sequence, defined recursively as
	\begin{equation}
		\begin{aligned}
			\bar E^{(k+1)}:&=\frac{\bar E^{(k)}+\mu^{(k)}q^{(k)}E^{(k+1)}}{1+\mu^{(k)}q^{(k)}},\\
			q^{(k+1)}:&=\mu^{(k)}q^{(k)}+1,
		\end{aligned}
		\label{eqn:recursion}
	\end{equation}
	with $\bar E^{(0)}:=E^{(0)}$, $q^{(0)}:=1$, and $\mu^{(k)}\in[0,1]$. The existence of $\alpha_{\atom}^{(k-1,\ell)}$ and $\alpha_{\latt}^{(k-1,\ell)}$ satisfying Eq. \eqref{eqn:nonmonotone termination criterion} is proved in Lemma \ref{lem:successfully generation} of Appendix \ref{appsec:convergence analysis}. It is not hard to verify that $\bar E^{(k)}$ is a weighted average of the past energies $\{E^{(j)}\}_{j=0}^k$ and is always larger than $E^{(k)}$. Hence, small energy increments are tolerable. That being said, the criterion \eqref{eqn:nonmonotone termination criterion} is sufficient to ensure the convergence of the entire procedure; see Theorem \ref{thm:convergence} in Appendix \ref{appsec:convergence analysis}.
	
	\par By virtue of the nonmonotone nature of the criterion \eqref{eqn:nonmonotone termination criterion}, we take the simple backtracking as the LM algorithm. Whenever the last trial configuration is rejected, we set
	\begin{equation}
		\alpha_{\atom}^{(k-1,\ell)}:=\alpha_{\atom}^{(k-1,\ell-1)}\delta_{\atom},~
		\alpha_{\latt}^{(k-1,\ell)}:=\alpha_{\latt}^{(k-1,\ell-1)}\delta_{\latt}
		\label{eqn:backtracking}
	\end{equation}
	where $\delta_{\atom}$, $\delta_{\latt}\in(0,1)$ are constants. Without carefully selected $\delta_{\atom}$ and $\delta_{\latt}$, we observe in most simulations that the ABB step sizes can directly meet the nonrestrictive criterion \eqref{eqn:nonmonotone termination criterion} and the rejected trials account for approximately 1.8\% of computational costs, in stark contrast to approximately 58.8\% in the widely used CG. 
	
	\subsection{Implementation Details}
	
	\par With the aforementioned developments, we have devised a projected gradient descent algorithm called PANBB for crystal structure relaxation under a fixed unit cell volume. More specifically, for the initial trial step sizes, we set $\alpha_{\atom}^{(-1,0)}=4.8\times10^{-2}$ \angstrom$^2$/eV, $\alpha_{\latt}^{(-1,0)}=10^{-6}$ \angstrom$^2$/eV during the first iteration \footnote{The specific initial step size for the atomic part is also used by CG realized in \textsc{vasp} and \textsc{cessp}. The numerical results of CG and PANBB will become more comparable. A small initial step size is chosen for the lattice part to prevent unexpected breakdown. In simulations, the performance of PANBB are not sensitive to the two values.} and compute in the subsequent iterations the truncated absolute values of the ABB step sizes:
	\begin{subequations}
		\begin{equation}
			\alpha_{\atom}^{(k-1,0)}=\max\lrbrace{\min\lrbrace{\abs{\alpha_{\atom,\abb}^{(k)}},\tau_{\atom}^{(k)},10},10^{-5}},
			\label{eqn:truncation atom}
		\end{equation}
		\begin{equation}
			\alpha_{\latt}^{(k-1,0)}=\max\lrbrace{\min\lrbrace{\abs{\alpha_{\latt,\abb}^{(k)}},\tau_{\latt}^{(k)},0.1},10^{-7}},
			\label{eqn:truncation latt}
		\end{equation}
		\label{eqn:truncation}
	\end{subequations}
	where 
	\begin{equation*}
		\begin{aligned}
			\tau_{\atom}^{(k)}:&=\gamma_{\atom}^{(k)}\max\lrbrace{-\log_{10}\big(\snorm{F_{\atom}^{(k)}}_\F/N\big),1},\\
			\tau_{\latt}^{(k)}:&=\gamma_{\latt}^{(k)}\max\lrbrace{-\log_{10}\big(\snorm{\tilde F_{\latt}^{(k)}}_\F/N\big),1}.
		\end{aligned}
	\end{equation*}
	The factors $\gamma_{\atom}^{(k)}$, $\gamma_{\latt}^{(k)}$ take respectively the values of 1, $10^{-3}$ for $k=0$ and later are modified according to the previous iterations: if they appear to be stringent from the recent history for $\alpha_{\atom}^{(k-1,0)}$ and $\alpha_{\latt}^{(k-1,0)}$ to pass the criterion \eqref{eqn:nonmonotone termination criterion}, we increase them; otherwise, we decrease them. For instance, 
	$$\gamma_{\atom}^{(k)}:=\left\{\begin{array}{ll}
		2\gamma_{\atom}^{(k-1)}, & \text{Case I};\\
		0.5\gamma_{\atom}^{(k-1)}, & \text{Case II};\\
		\gamma_{\atom}^{(k-1)}, & \text{otherwise},
	\end{array}\right.$$
	where
	\begin{itemize}
		\item {\bf Case I:} the truncation has been active but $\alpha_{\atom}^{(k-1,0)}$ has directly passed the criterion \eqref{eqn:nonmonotone termination criterion} for twice in the past $\min\{k-\hat k,20\}$ iterations;
		
		\item {\bf Case II:} $\alpha_{\atom}^{(k-1,0)}$ has failed to pass the criterion \eqref{eqn:nonmonotone termination criterion} for twice in the past $\min\{k-\hat k,20\}$ iterations.
	\end{itemize}
	Initially, we set $\hat k:=0$; whenever the factor $\gamma_{\atom}^{(k)}$ gets changed, we set $\hat k:=k$. The update on $\gamma_{\latt}^{(k)}$ follows similar rules. Note that the box boundaries in Eq. \eqref{eqn:truncation}, such as $10^{-5}$ and 10 for $\alpha_{\atom}^{(k-1,0)}$, are present for the ease of theoretical analyses. In practice, these boundaries are rarely touched. 
	The criterion \eqref{eqn:nonmonotone termination criterion} is equipped with $\eta=10^{-4}$ and $\mu^{(k)}\equiv0.05$ following \cite{hu2022force}. We take $\delta_{\atom}=0.1$ and $\delta_{\latt}=0.5$ for the backtracking \eqref{eqn:backtracking} \footnote{We use a small $\delta_{\atom}$ to end the LM process more quickly and help the ensuing ABB step sizes learn local curvatures better. The value of $\delta_{\latt}$ is not so critical, since usually $\alpha_{\latt}^{(k-1,0)}$ is already sufficiently small.}. Under mild conditions, we establish the convergence of PANBB to equilibrium states; see Theorem \ref{thm:convergence} in Appendix \ref{appsec:convergence analysis}. 
	
	\section{Numerical Simulations}
	
	\subsection{Experimental Settings}
	
	\par We have implemented PANBB in the in-house plane-wave code \textsc{cessp} \cite{fang2016on,gao2017parallel,zhou2018applicability,fang2019implementation}. The exchange-correlation energy is described by the generalized gradient approximation \cite{perdew1996generalized}. Electron-ion interactions are treated with the projector augmented-wave (PAW) potentials based on the open-source \textsc{abinit} Jollet-Torrent-Holzwarth data set library in the PAW-XML format \cite{jollet2014generation}. The total energies are calculated using the Monkhorst-Pack mesh \cite{monkhorst1976special} with a $k$-mesh spacing of $0.12\sim0.20$ \angstrom$^{-1}$. The typical plane-wave energy cutoffs are around $500\sim600$ eV. The KS equations are solved by the preconditioned self-consistent field (SCF) iteration \cite{zhou2018applicability}. 
	
	\begin{figure*}[!t]
		\centering
		\includegraphics[width=.45\linewidth]{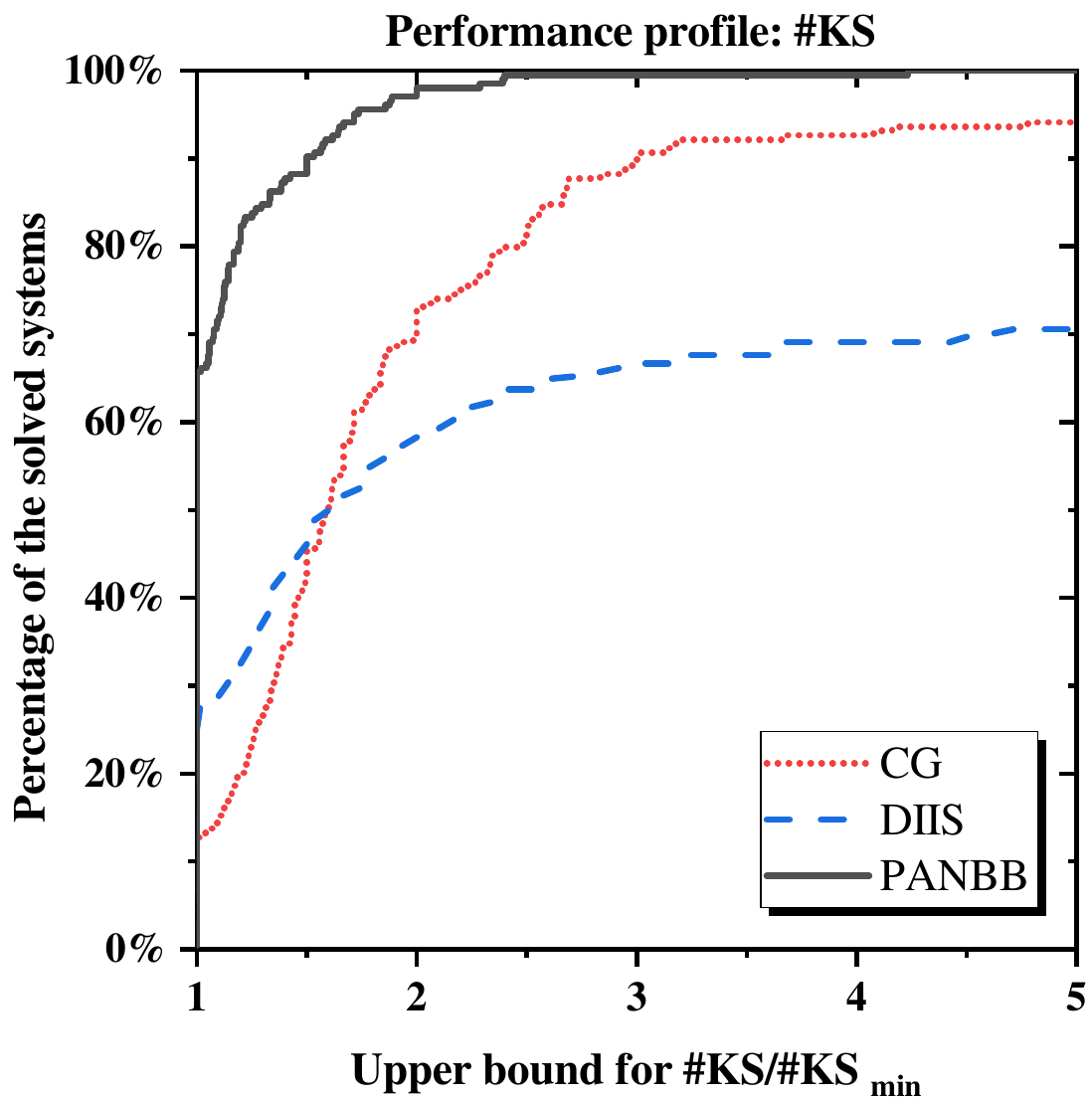}\quad\quad
		\includegraphics[width=.45\linewidth]{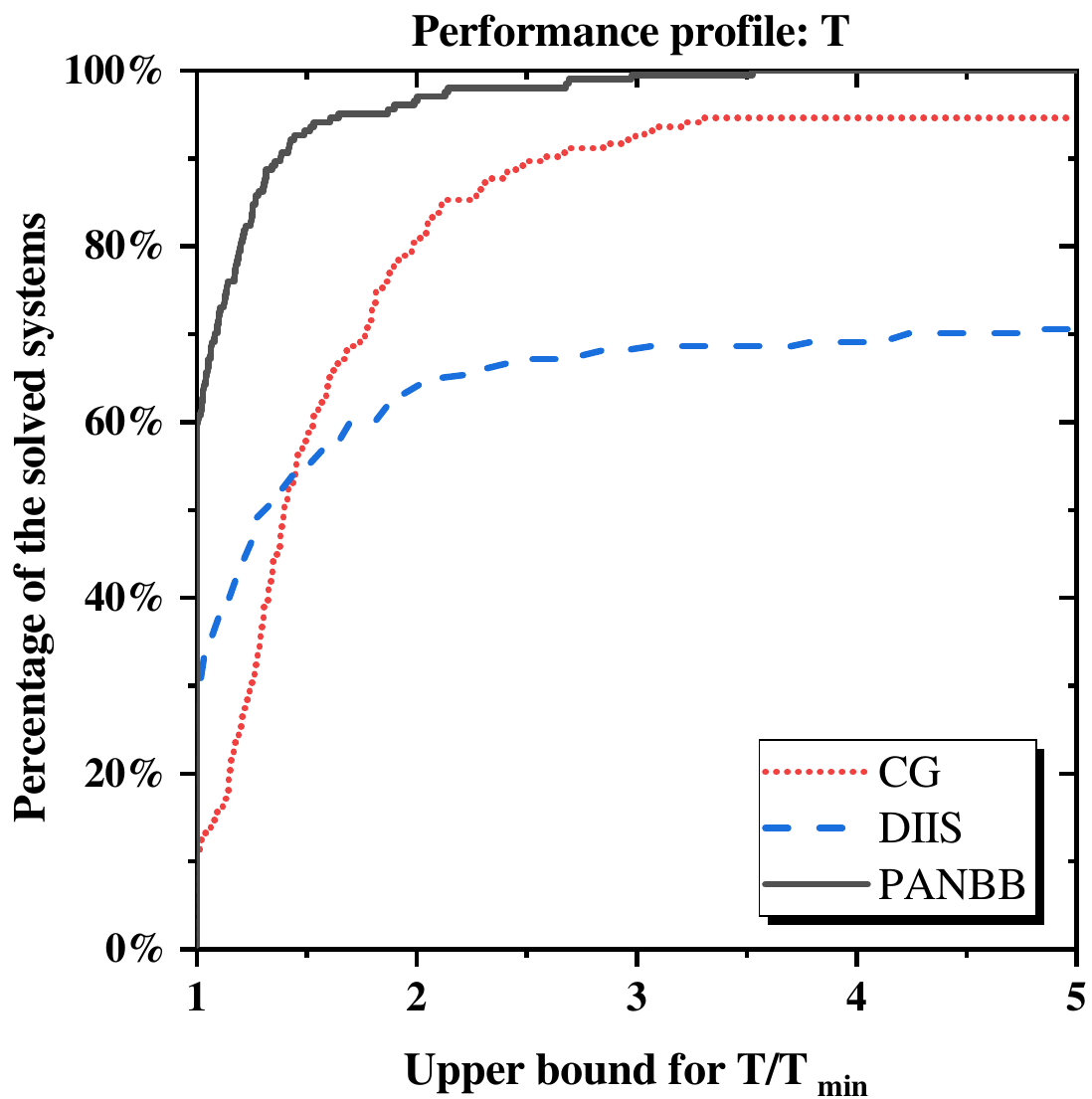}
		\caption{The performance profiles of CG, DIIS, and PANBB, where \#KS$_{\min}$ and T$_{\min}$ denote the minimum \#KS and T required by the three methods on a given system, respectively. To reach a fair comparison and illustrate the algorithm robustness at the same time, we only retain (i) the systems where the difference between the converged energies obtained by each pair of methods does not exceed 3 meV per atom; (ii) the systems on which at least one method diverges or fails to converge before \#KS exceeds 1000, and set the associated \#KS and T to $\infty$. The final number of the systems used for this figure is 204. Left: \#KS. Right: T. } 
		\label{fig:perf}
	\end{figure*}
	
	\par In addition to PANBB, we test the performance of the CG and DIIS built in \textsc{cessp}. Notably, compared with the one documented in textbook, the CG in test uses deviatoric stress tensors in search direction updates for the fixed-volume relaxation. We terminate all the three methods whenever the number of solving the Kohn-Sham equations (abbreviated as \#KS) reaches 1000 or the maximum atomic force and the maximum stress components divided by $N$ fall below 0.01 eV/\angstrom. The convergence criterion for the SCF iteration is $10^{-5}$ eV. To evaluate the performance of the three methods, we record the \#KS and running time in seconds (abbreviated as T) required for fulfilling the force tolerance.
	
	\subsection{Benchmark Test}
	
	\begin{figure*}[!t]
		\centering
		\includegraphics[width=.47\linewidth]{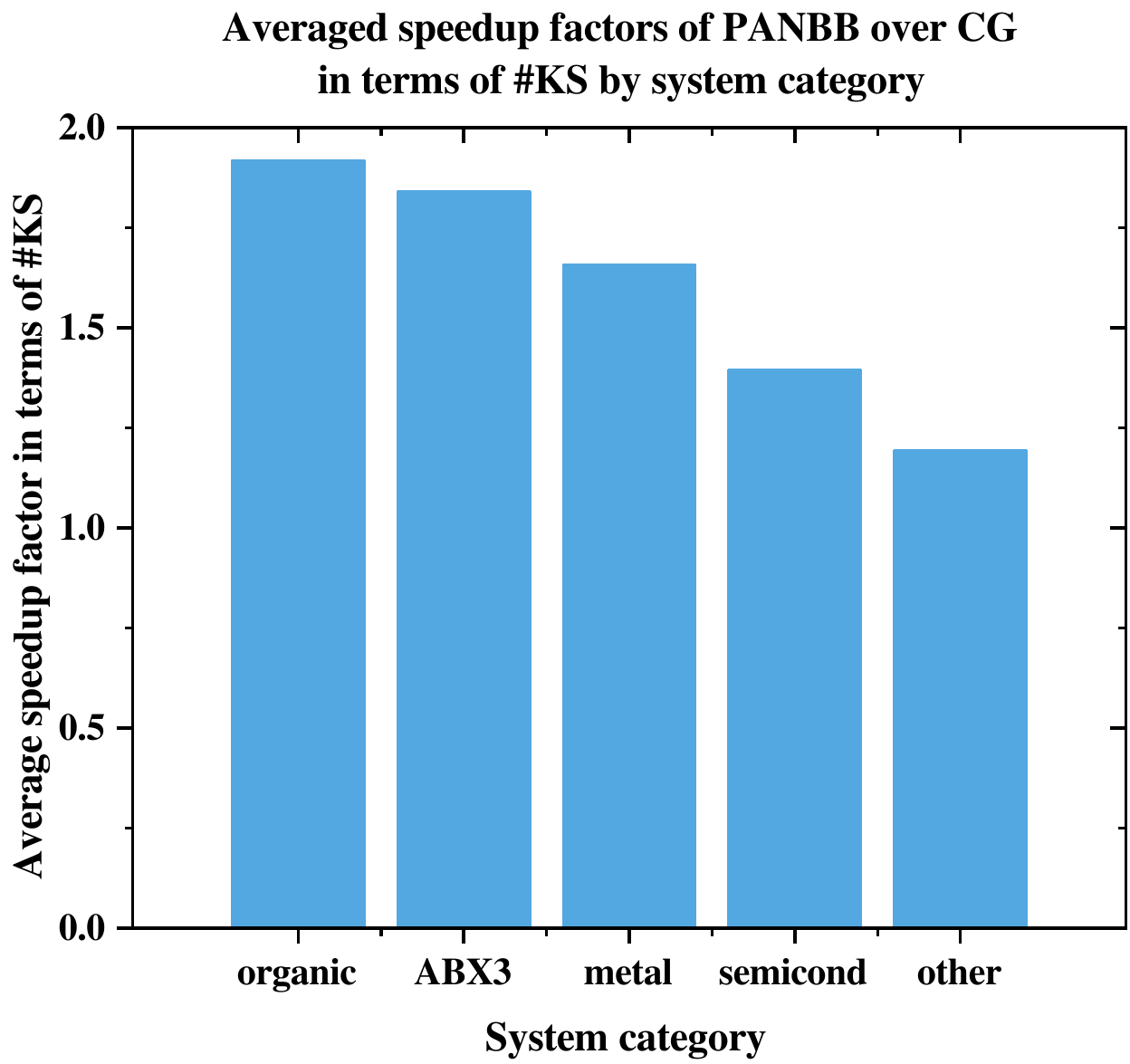}\quad\quad
		\includegraphics[width=.47\linewidth]{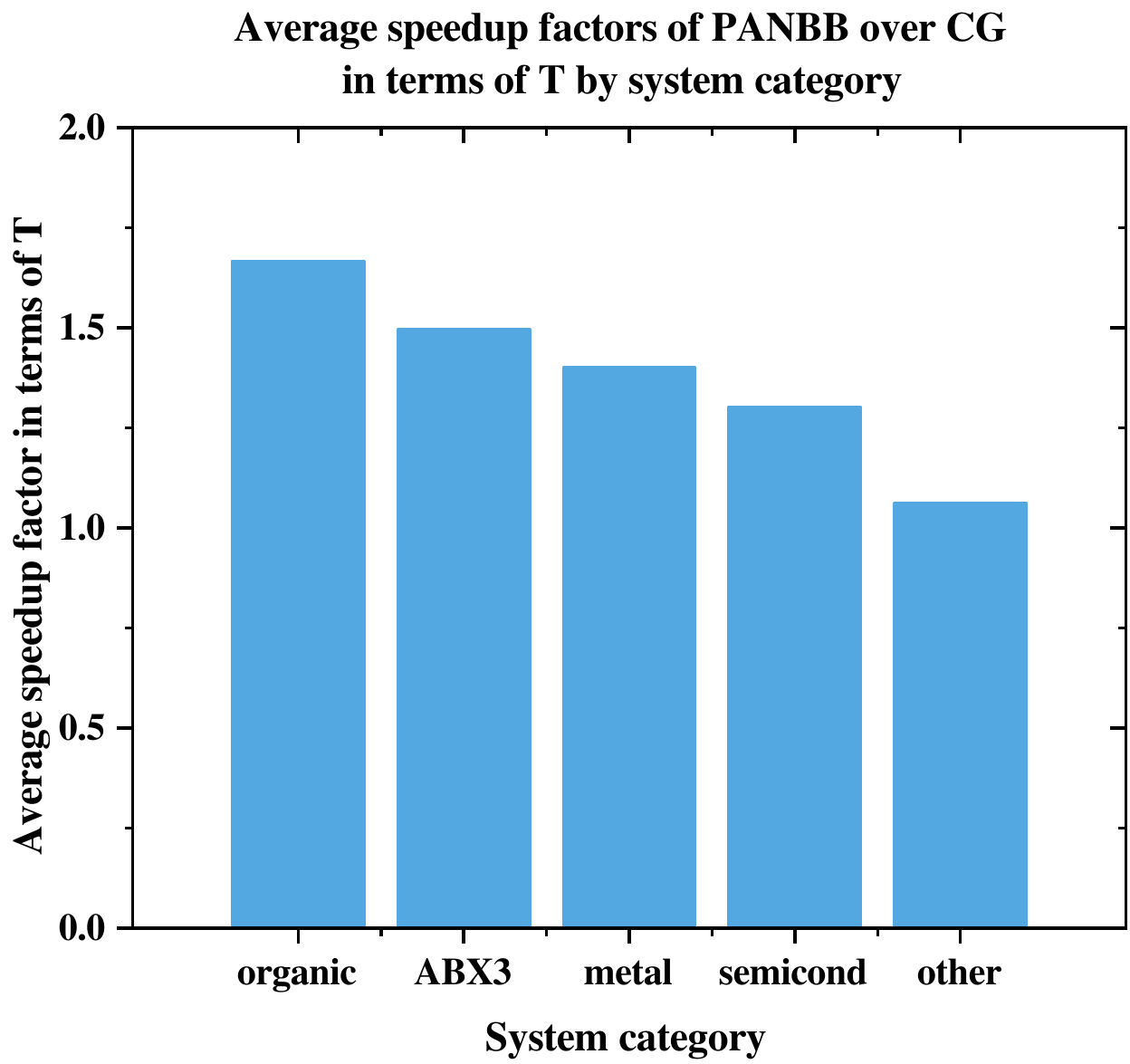}
		\caption{The average speedup factors of PANBB over CG by the system category. To achieve a fair comparison, we only include the systems where both PANBB and CG converge normally and the converged energy differences per atom are not larger than 3 meV. Categories whose numbers of the included systems are less than 5 are merged together into ``other''. Left: \#KS. Right: T.}
		\label{fig:speedup category}
	\end{figure*}
	
	\par We have performed numerical comparisons among CG, DIIS, and PANBB on a benchmark set consisting of 223 systems from various categories, such as metallic systems, organic molecules, semiconductors, $ABX_3$ perovskites, heterostructures, and 2D materials. About 68.6\% of them are metallic systems, for which the EOS calculations are more demanding. The number of atoms in one system ranges from 2 to 215. Some of them are available from the Materials Project \cite{Jain2013} and the Organic Materials Database \cite{borysov2017organic}. In the beginning, these systems are at their ideal configurations for defects, heterostructures, and substitutional alloys, or simply place a molecule on top of a surface. More information about the benchmark set can be found in the Supplemental Material. 
	
	\par To provide a comprehensive comparison among CG, DIIS, and PANBB over the benchmark set, we adopt the performance profile \cite{dolan2002benchmarking} as shown in Fig. \ref{fig:perf}, where \#KS$_{\min}$ and T$_{\min}$ represent the minimum \#KS and T required by the three methods on a given system, respectively. For example, in the right panel of Fig. \ref{fig:perf}, the intercepts of the three curves indicate that CG, DIIS, and PANBB are the fastest on approximately 11.3\%, 28.9\%, and 59.8\% of all the systems, respectively. The $y$-coordinate touched by the curve of PANBB when the $x$-coordinate equals two shows that the running time of PANBB is not twice larger than those of the other two methods on approximately 96.5\% of all the systems. Basically, a larger area below the curve implies better overall performance of the associated method. We can therefore conclude from Fig. \ref{fig:perf} that PANBB favors the best overall performance. 
	In addition, the Supplemental Material reveals that PANBB is respectively faster than CG and DIIS on approximately 85.2\% and 68.8\% of all the systems. The average speedup factors of PANBB over CG and DIIS, in terms of running time, are around 1.41 and 1.45, respectively. We also calculate the average speedup factors of PANBB over CG by the system category; see Fig. \ref{fig:speedup category} for an illustration. Among others, the acceleration brought by PANBB is the most evident when relaxing organic molecules, $ABX_3$ perovskites, and metallic systems, with speedup factors of 1.67, 1.50, and 1.40, respectively, in terms of running time. These statistics provide strong evidence for the universal efficiency superiority of PANBB. The robustness of PANBB also deserves a whistle. Thanks to its theoretical convergence guarantees (see Theorem \ref{thm:convergence} in Appendix \ref{appsec:convergence analysis}), PANBB converges consistently across the benchmark set, whereas CG fails on 11 systems due to the breakdown of the LM algorithm, and DIIS diverges on 56 systems.
	
	\subsection{Calculations of EOS on the AlCoCrFeNi HEA}
	
	\par As an application of crystal structure relaxation under a fixed unit cell volume, we employ PANBB to calculate the static EOS up to approximately 30 GPa for the body-centered cubic AlCoCrFeNi HEA. The local lattice distortion (LLD), induced mainly by the large atomic size mismatch of the alloy components, is one of
	the core effects responsible for the unprecedented mechanical
	behaviors of HEAs. The LLD in some equimolar HEAs under ambient pressure is studied by \ab~approach \cite{song2017local}. The volume-dependent LLD is investigated and is helpful for understanding the structural properties of HEAs under pressure. The thermodynamic properties have also been investigated in comparison with experimental data for validation. 
	
	\begin{figure*}[!t]
		\centering
		\includegraphics[trim=6mm 6mm 6mm 6mm, width=.4\columnwidth]{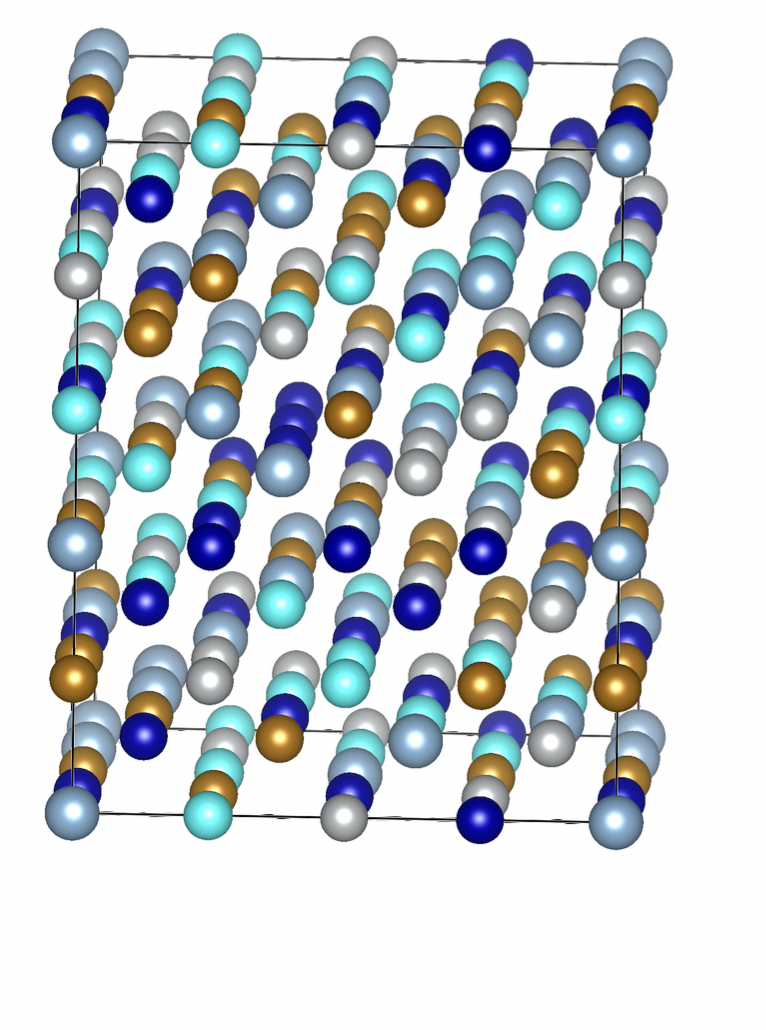}
		\includegraphics[trim=6mm 6mm 6mm 6mm, width=.4\columnwidth]{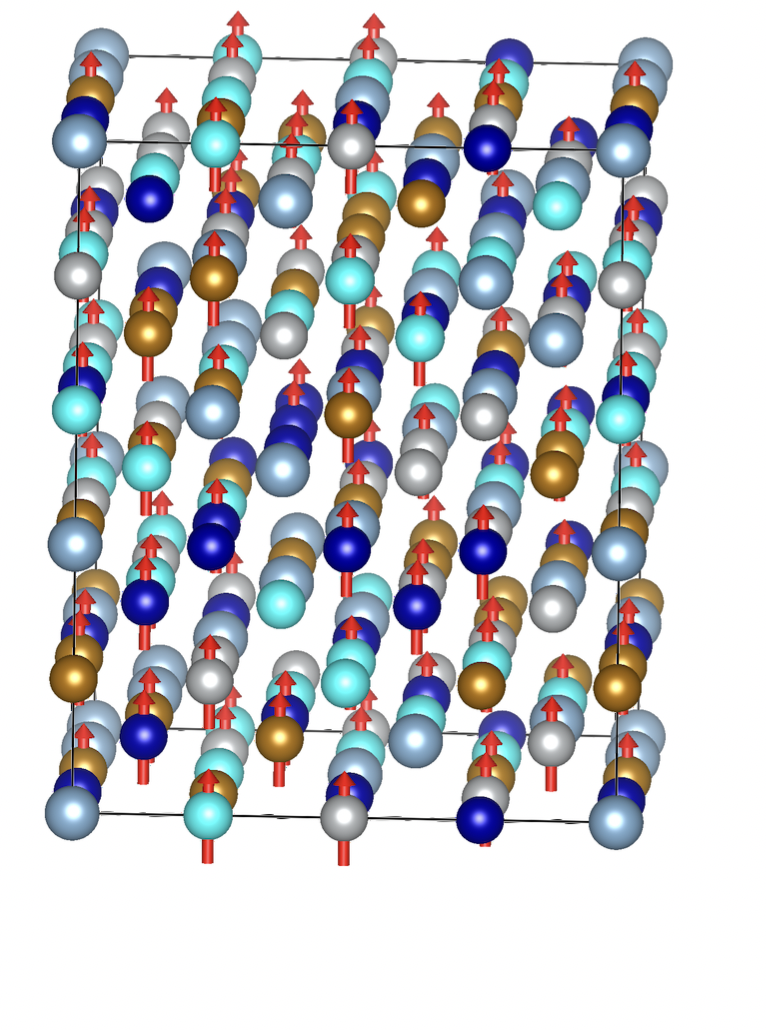}
		\includegraphics[trim=6mm 6mm 6mm 6mm, width=.4\columnwidth]{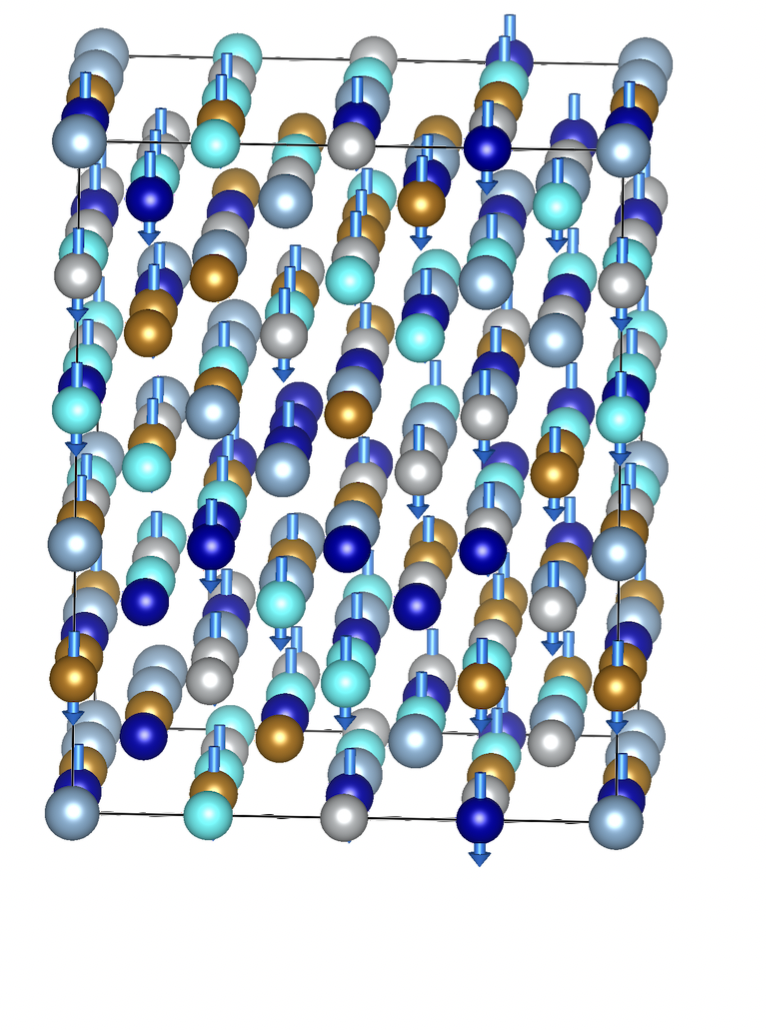}
		\includegraphics[trim=0mm 0mm 0mm 0mm, width=.65\columnwidth]{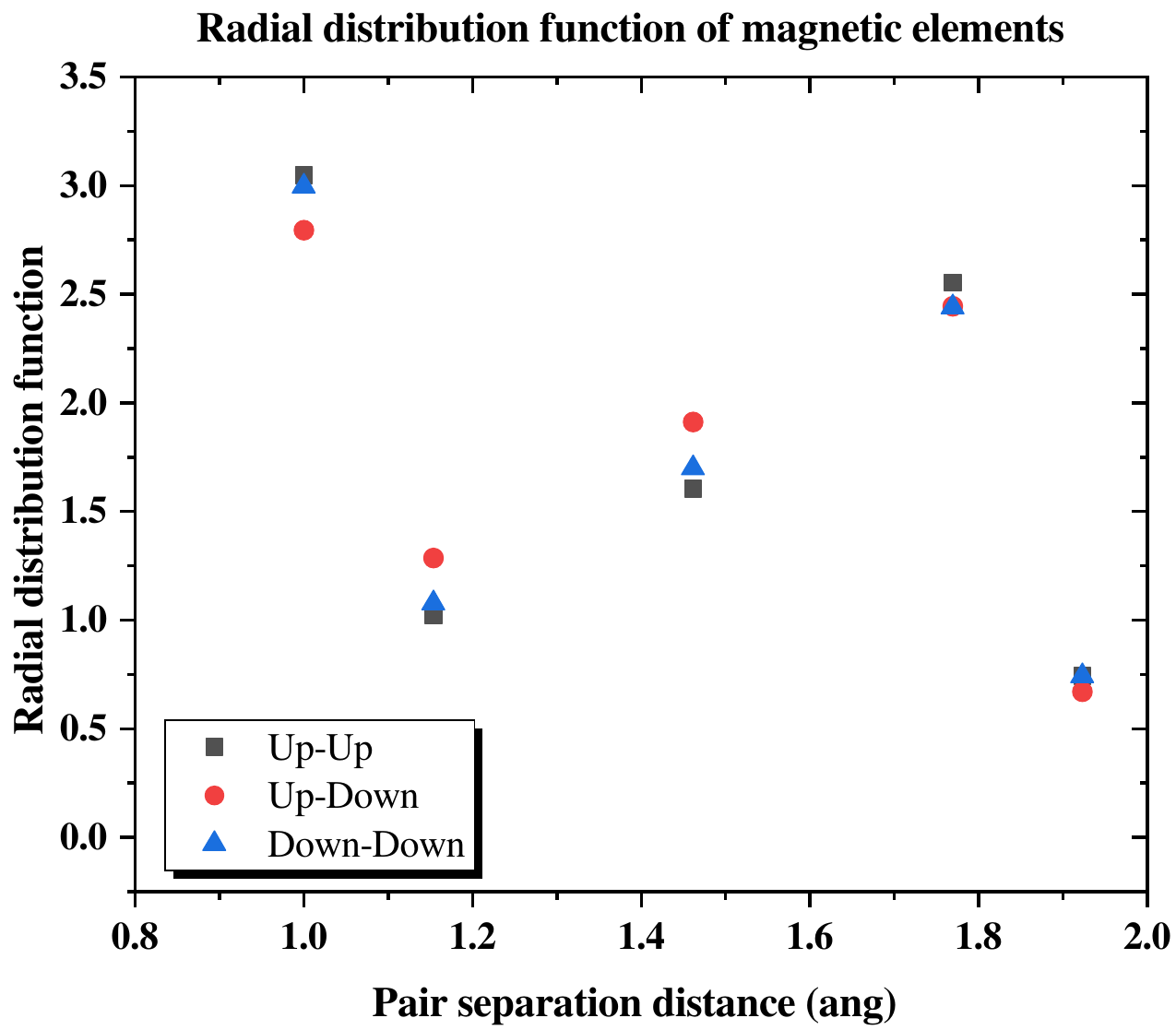}
		\caption{The generated atomic and magnetic configurations. From left to right: the atomic configuration generated by the SAE method, the 64 spin-up and 64 spin-down sites chosen by the integer programming approach for magnetic configuration, and the radial distribution function of the magnetic elements.}
		\label{fig:modeling}
	\end{figure*}
	
	\par To incorporate the temperature-induced magnetic structural transition, we need to describe both the ferromagnetic (FM) and paramagnetic (PM) states of AlCoCrFeNi. To characterize the inherent chemical and magnetic disorder, we combine the SAE method \cite{tian2020structural} and an integer programming approach \cite{conforti2014integer} to construct representative $4\times4\times5$ supercells with 160 atoms. Specifically, the SAE method is first invoked for chemical disorder modeling, where we seek to minimize 
	$$\sum_{d(C_2)<r_2^c}w_2f(C_2,\sigma)+\sum_{d(C_3)<r_3^c}w_3f(C_3,\sigma)$$
	with respect to the atomic occupation $\sigma$. Here $w_n$ ($n=2,3$) are appropriate weights and $r_n^c$ ($n=2,3$) are cutoff radii. The notations $f(C_2,\sigma)$ and $f(C_3,\sigma)$ represent the similarity functions associated with the diatomic clusters $C_2$ and triatomic clusters $C_3$, respectively. The corresponding diameter of the cluster is denoted by $d(C_2)$ or $d(C_3)$. We perform 90,000 Monte carlo samplings and minimize $f(C_2,\sigma)$ to approximately 0.06 and $f(C_3,\sigma)$ to approximately 0.18. 
	
	\par For magnetic disorder modeling, we refrain from using the SAE method again because it necessitates distinguishing spin-up and spin-down elements, leading to a higher-dimensional search space. Instead, we resort to an integer programming approach to distribute spin-up and spin-down sites uniformly in the configuration given by the SAE method, as a way to achieve the disordered local moment approximation \cite{gyorffy1985first}. We refer readers to Appendix \ref{appsec:modeling details} for more technical details. The generated atomic and magnetic configurations are illustrated in Fig. \ref{fig:modeling}. The rightmost panel displays the radial distribution function of the magnetic elements, demonstrating a proper description of the magnetic disorder by the integer programming approach.
	
	\par To calculate the static EOS for both the FM and PM states, we distribute 18 unit cell volumes in the range between 10.29 \angstrom$^3$/atom and 13.12 \angstrom$^3$/atom, correponding to the static pressures from approximately $29.4$ GPa to $-12.9$ GPa. For the FM state, we employ PANBB for the fixed-volume relaxation with colinear calculations. Following \cite{wu2020structural}, the PM state is then described based on the relaxed FM structures, with the initial magnetic configuration determined by the integer programming approach. The energy-volume relations on both the FM and PM AlCoCrFeNi are depicted in Fig. \ref{fig:EV FM PM AlCoCrFeNi}, 
	\begin{figure}[htb]
		\centering
		\includegraphics[width=\columnwidth]{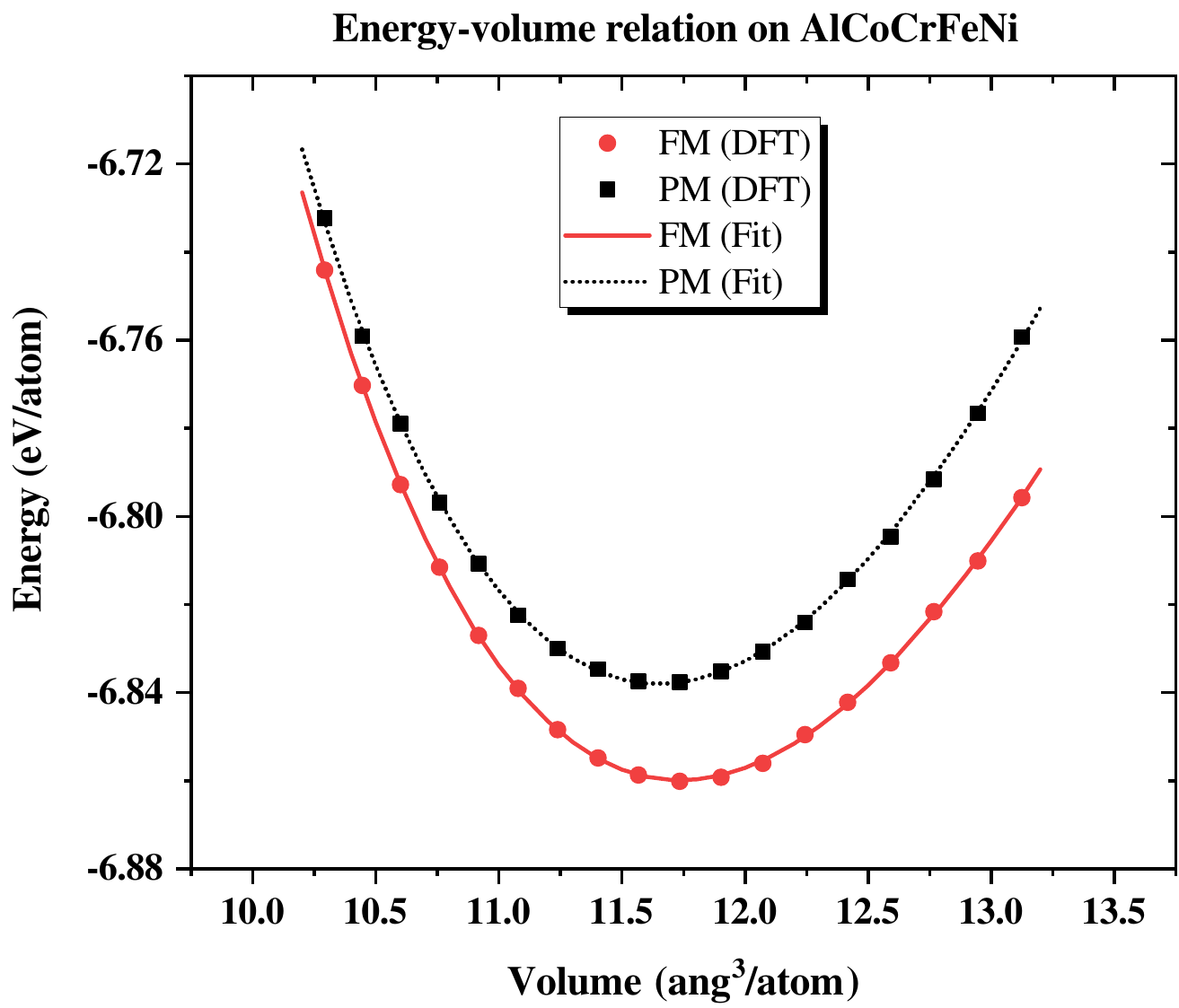}
		\caption{The energy-volume relations on FM and PM AlCoCrFeNi. ``DFT'' indicates the \ab~results, while ``Fit'' is for the curve fitted to the static BM3 EOS. The red circles and solid line represent the results for the FM state, while the black squares and dotted line represent the results for the PM state.}
		\label{fig:EV FM PM AlCoCrFeNi}
	\end{figure}
	along with the fitted third-order Birch-Murnaghan (BM3) EOS \cite{birch1947finite, otero2011gibbs2}. We record the fitted parameters in Table \ref{tab:fitted params}, where $V_0$ represents the equilibrium volume and $E_0$, $B_0$, and $B_0'$ denotes the energy, bulk modulus, and its derivative with respect to pressure associated with $V_0$, respectively.
	
	\begin{table}[htb]
		\centering
		\caption{The fitted parameters of the static BM3 EOS.}
		\label{tab:fitted params}
		\begin{tabular}{l|cccc}
			\toprule
			\multicolumn{1}{c|}{References} & $V_0$ (\angstrom$^3$/atom) & $E_0$ (eV/atom) & $B_0$ (GPa) & $B_0'$ \\\midrule
			FM (\cite{wu2020structural}) & 11.77 & -6.85 & 157.32 & 8.9 \\
			
			FM (PANBB) & 11.74 & -6.86 & 158.49 & 5.3 \\
			
			FM (CG) & 11.73 & -6.86 & 158.04 & 5.8 \\ \midrule
			
			PM (\cite{wu2020structural}) & 11.70 & -6.82 & 161.64 & 5.3 \\
			
			PM (PANBB) & 11.65 & -6.84 & 167.59 & 4.7 \\
			
			PM (CG) & 11.65 & -6.84 & 161.64 & 4.0 \\
			\bottomrule
		\end{tabular}
	\end{table}
	
	\par Compared with the previous calculations \cite{wu2020structural}, our fitted $V_0$ and $B_0$ for both the FM and PM states, $E_0$ for the FM state, and $B_0'$ for the PM state are comparable; see Table \ref{tab:fitted params}. The gaps in $B_0'$ for the FM state and $E_0$ for the PM state can result from two aspects: (i) we enhance magnetic disorder by the integer programming approach; (ii) \cite{wu2020structural} relaxes only the atomic positions, while PANBB optimizes both the atomic and lattice degrees of freedom. By the way, our value of $B_0'$ appears to be more reasonable since it is closer to those obtained for similar compositions \cite{song2017local}. Within the mean-field approximation \cite{ge2018effect}, we can estimate the Curie temperature of AlCoCrFeNi as 
	$$\frac{2}{3(1-c)k_B}(E_0^{\text{PM}}-E_0^{\text{FM}})\approx210.8~\text{K},$$
	where $c$ is the concentration of nonmagnetic atoms, $k_B$ is the Boltzmann's constant, $E_0^{\text{FM}}$ and $E_0^{\text{PM}}$ are the energies of the FM and PM states associated with the equilibrium volumes from the BM3 EOS fitting based on the PANBB results, respectively; see Fig. \ref{fig:EV FM PM AlCoCrFeNi}. Our estimated Curie temperature lies in the interval of the critical temperature of Al$_x$CoCrFeNi \cite{kao2011electrical}, which is between 200 K ($x=0.75$) and 400 K ($x=1.25$), to some extent demonstrating rational descriptions of the FM and PM states. 
	
	\par To construct the EOS in the temperature ($T$) and pressure ($P$) space, the modified mean field potential (MMFP) approach \cite{song2007modified} is used to describe the ion vibration energy as 
	\begin{align*}
		F_{\text{vib}}^{\text{MMFP}}:&=-k_BT\lrsquare{\frac{3}{2}\ln\frac{Mk_BT}{2\pi\hbar}+\ln v_f(V,T)},\\
		v_f(V,T):&=4\pi\int\exp\lrsquare{\frac{p(r,V)}{k_BT}}r^2\dd r,\\
		p(r,V):&=\frac{1}{2}[E_s(D+r)+E_s(D-r)-2E_s(D)]\\
		&\quad +\frac{\lambda r}{2D}[E_s(D+r)-E_s(D-r)],
	\end{align*}
	where $\hbar$ denotes the Planck constant, $M$ is the relative atomic mass, $D$ is the nearest atomic distance, $r$ corresponds to the distance of the ion vibration away from its equilibrium position, $E_s$ is the static EOS, and $p$ is the constructed potential. Here, $\lambda$ can take three values, $-1$, $0$, $1$, which originate from different forms of the Gr\"uneisen coefficient. As a powerful tool to calculate the Helmholtz free energy, the MMFP effectively describes the ion vibration with the anharmonic effect, and has been widely used in thermodynamic physical properties simulations of HEAs \cite{tian2020structural,wu2020structural}.
	
	\par Since the effectiveness of the MMFP approach highly depends on the potential field constructed from the static EOS, we can validate the previous calculations by investigating the thermodynamic properties of AlCoCrFeNi. We first calculate the equilibrium volume and bulk modulus at ambient conditions with the MMFP approach using the PM PANBB results. Our values (approximately 11.83 \AA$^3$/atom and 148.8 GPa) turn out to be close to those from the X-ray diffraction and diamond anvil cell experiments (11.79 \AA$^3$/atom and 150$\pm$2.5 GPa) \cite{cheng2019pressure}. We then calculate the volume-pressure relation at 300 K with the PM PANBB results; see the solid line with label ``300 K (MMFP)'' in Fig. \ref{fig:V-P}. 
	At equivalent pressures, the maximum relative deviation of the volumes from the diamond anvil cell data \cite{cheng2019pressure} does not exceed 0.7\%.
	
	\begin{figure}[htb]
		\centering
		\includegraphics[width=\columnwidth]{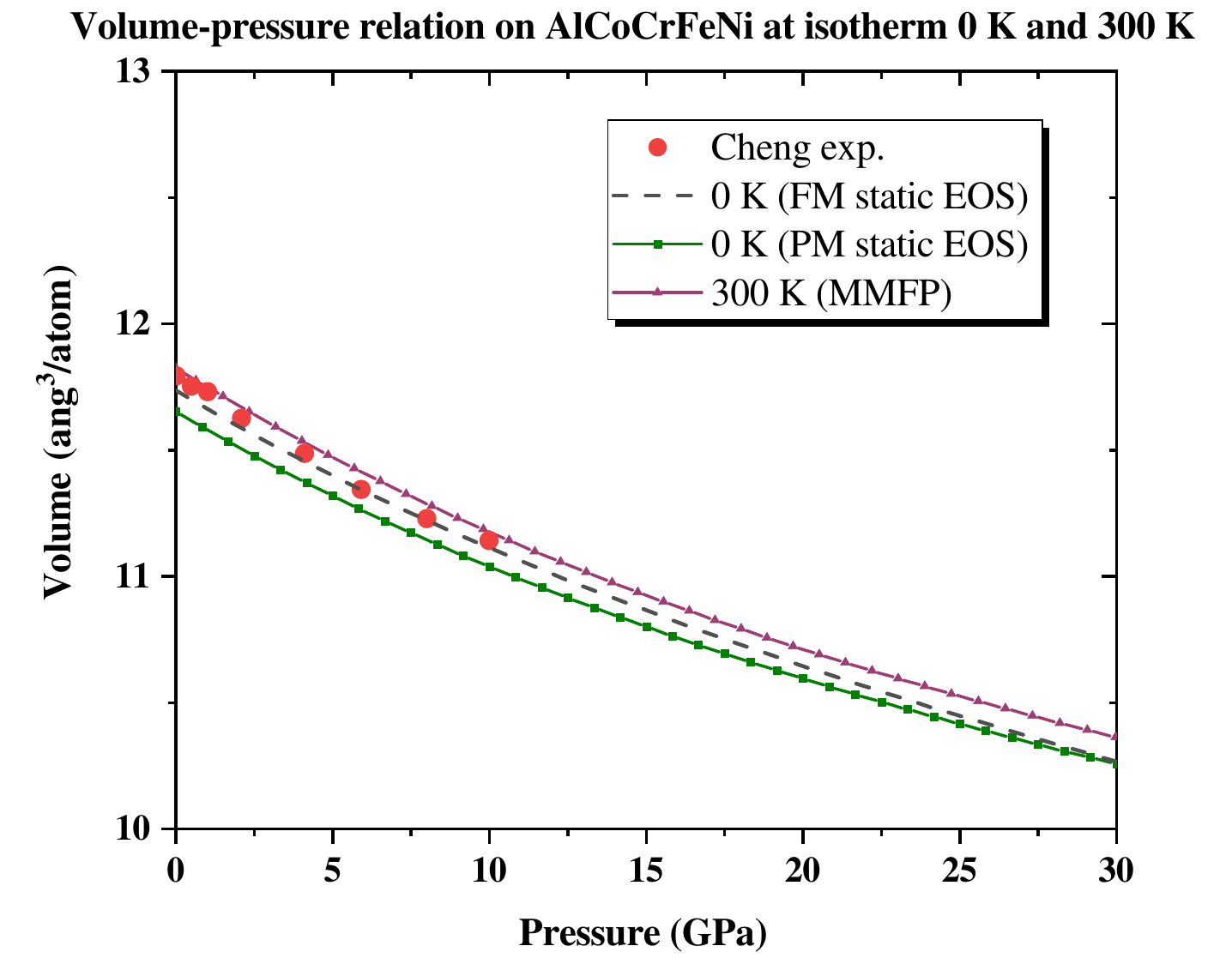}
		\caption{The volume-pressure relation on AlCoCrFeNi at 0 K and 300 K. ``Cheng exp.'' represents the X-ray diffraction and diamond anvil cell experimental results \cite{cheng2019pressure}, ``0 K (FM static EOS)'' refers to the relation by fitting the static EOS based on the FM PANBB results, ``0 K (PM static EOS)'' refers to the relation by fitting the static EOS based on the PM PANBB results, and ``300 K (MMFP)'' refers to the relation given by the MMFP approach based on the PM PANBB results.}
		\label{fig:V-P}
	\end{figure}
	
	\par We would also like to compare PANBB with CG on the relaxation of the FM AlCoCrFeNi. PANBB converges normally at all volumes, while CG is found to fail at 4 volumes due to the breakdown of the LM process. Moreover, a sharp jump of $c/a$ is found during the relaxation of CG when the pressure exceeds 21 GPa; see Fig. \ref{fig:axislength ratio}.
	\begin{figure}[htb]
		\centering
		\includegraphics[width=\columnwidth]{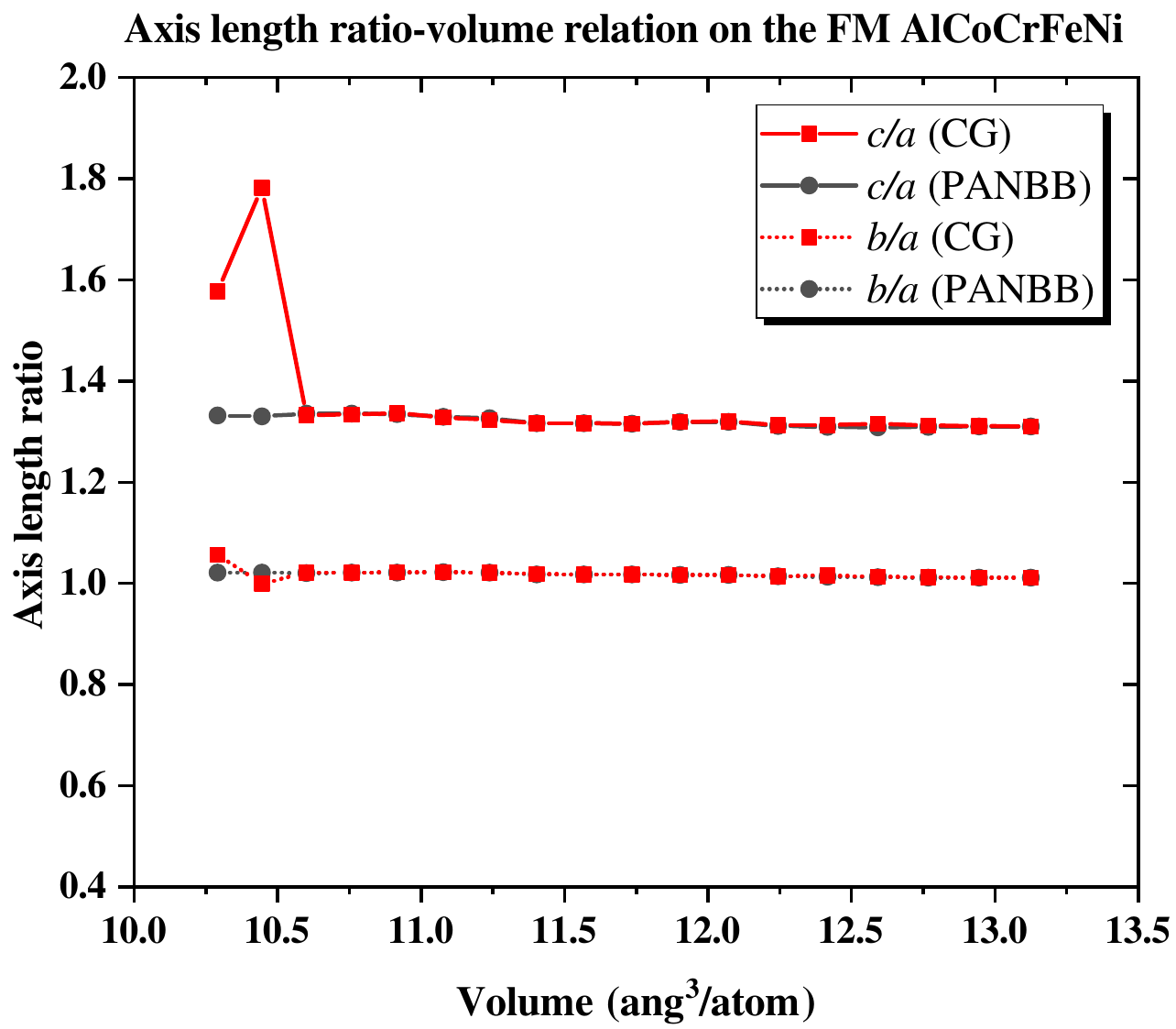}
		\caption{The axis length ratio-volume relation on the FM AlCoCrFeNi. For the volumes where CG fails to converge normally, we use the axis length ratios of the final structures. The axis length ratios ``$c/a$'' and ``$b/a$'' are indicated by solid and dotted lines, respectively. The results of CG and PANBB are indicated by square and circle markers, respectively.}
		\label{fig:axislength ratio}
	\end{figure}
	The results of CG at the other 16 volumes are fitted by the BM3 EOS. The fitted parameters of the static EOS of AlCoCrFeNi are summarized in Table \ref{tab:fitted params} and consistent with those derived from the PANBB results. The average speedup factor of PANBB over CG in terms of running time is 1.36. Note that this speedup factor is calculated at the volumes where both methods converge normally and the energy differences per atom are not larger than 3 meV. 
	
	\par By the way, we illustrate the LLD achieved by fixed-volume relaxation at pressures of approximately 0 GPa, 21 GPa, and 29 GPa using the smearing radial distribution function. This stands in sharp contrast to the peaks at coordination shells in the ideal lattice, as depicted in Fig. \ref{fig:rdf}. Notably, when the pressure does not exceed 21 GPa, the radial distribution functions of the structures relaxed by PANBB and CG display a high degree of consistency. However, as the pressure rises to 29 GPa, the radial distribution function obtained by PANBB exhibits smeared peaks consistent with those of the ideal lattice, while that obtained by CG shows significant displacements. This can be attributed to the previously mentioned issue of strong lattice distortion during the relaxation process of CG. Whether such lattice distortion conforms to physics deserves future investigations.
	\begin{figure*}[!t]
		\centering
		\includegraphics[width=.33\linewidth]{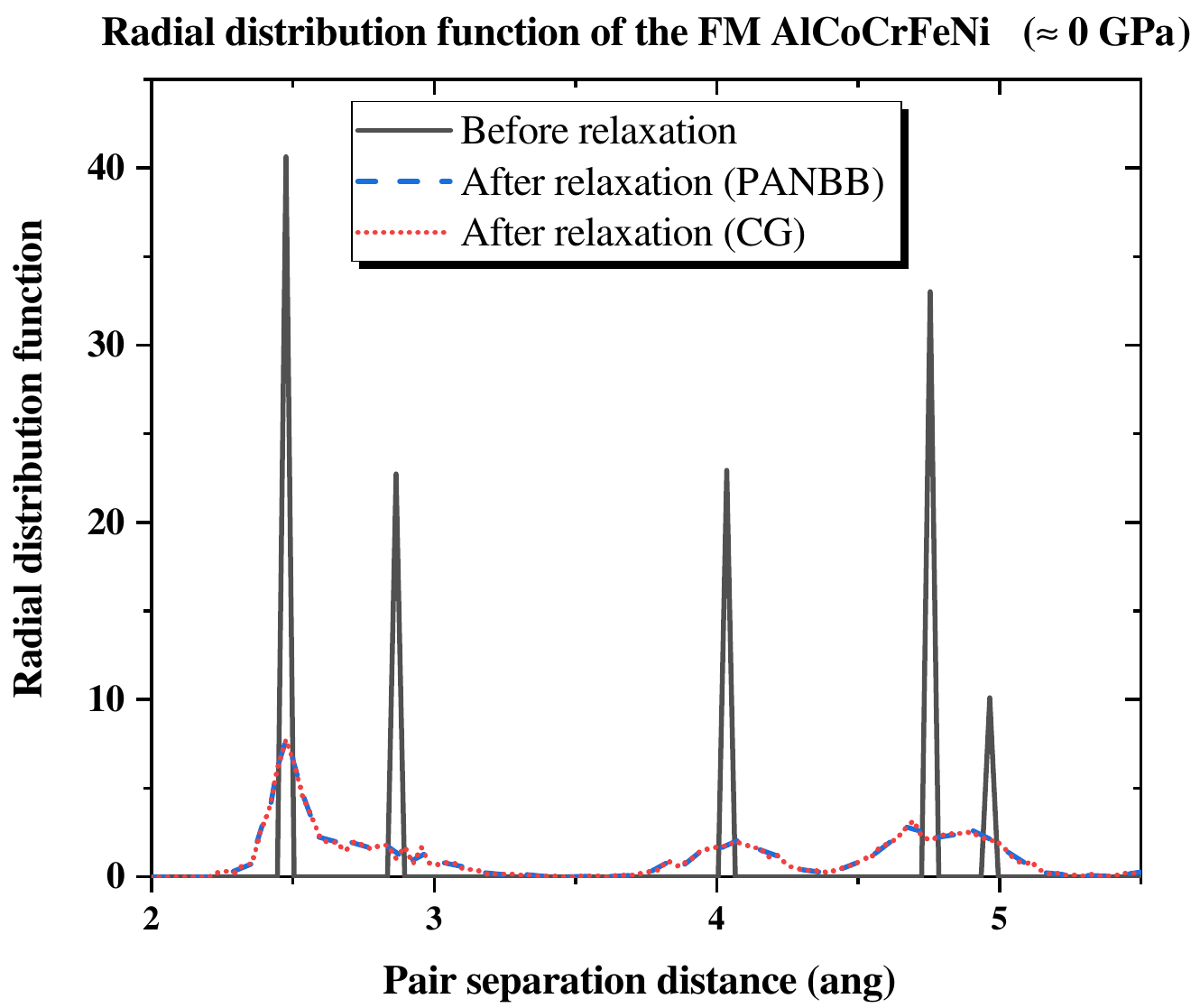}
		\includegraphics[width=.33\linewidth]{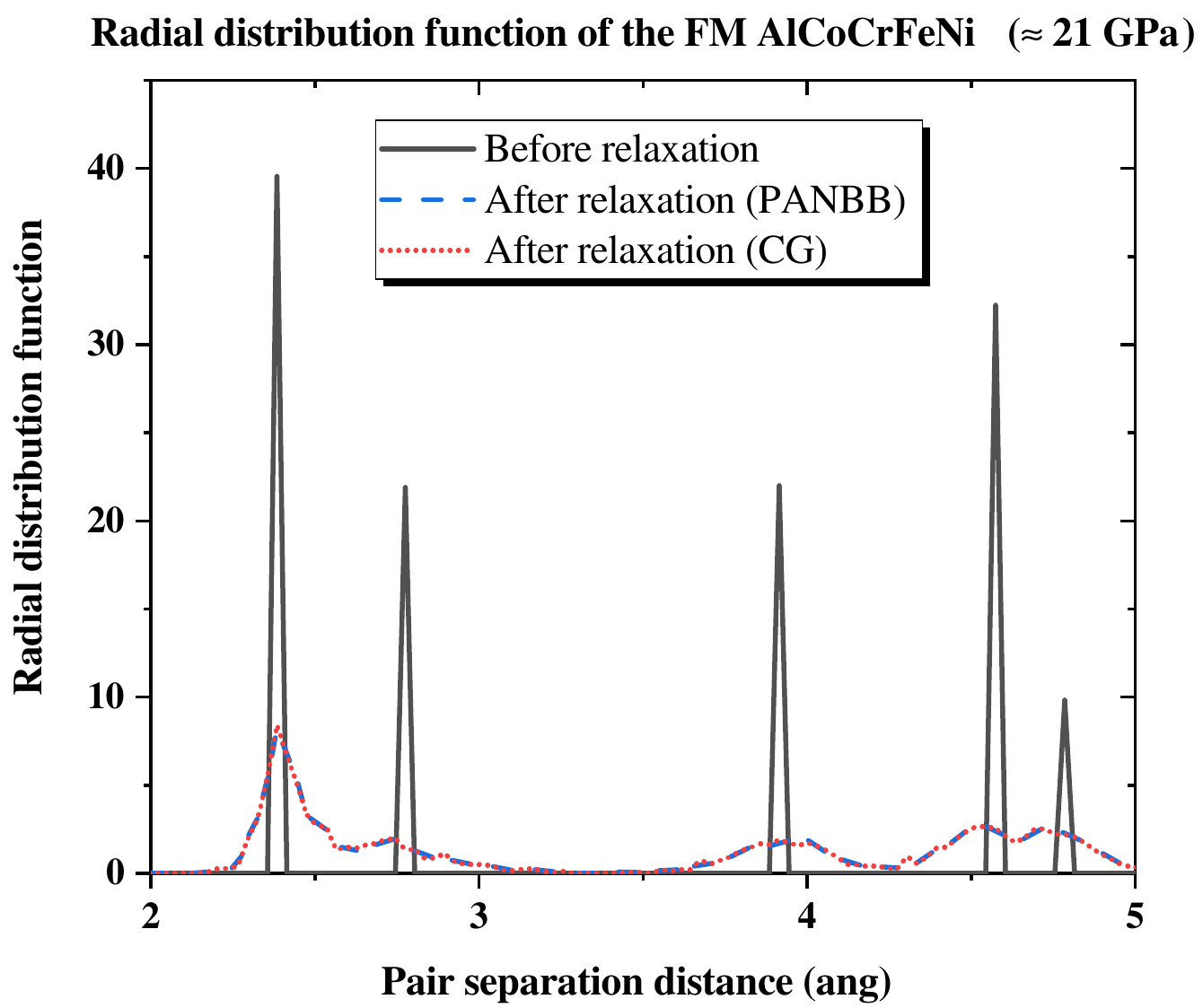}
		\includegraphics[width=.33\linewidth]{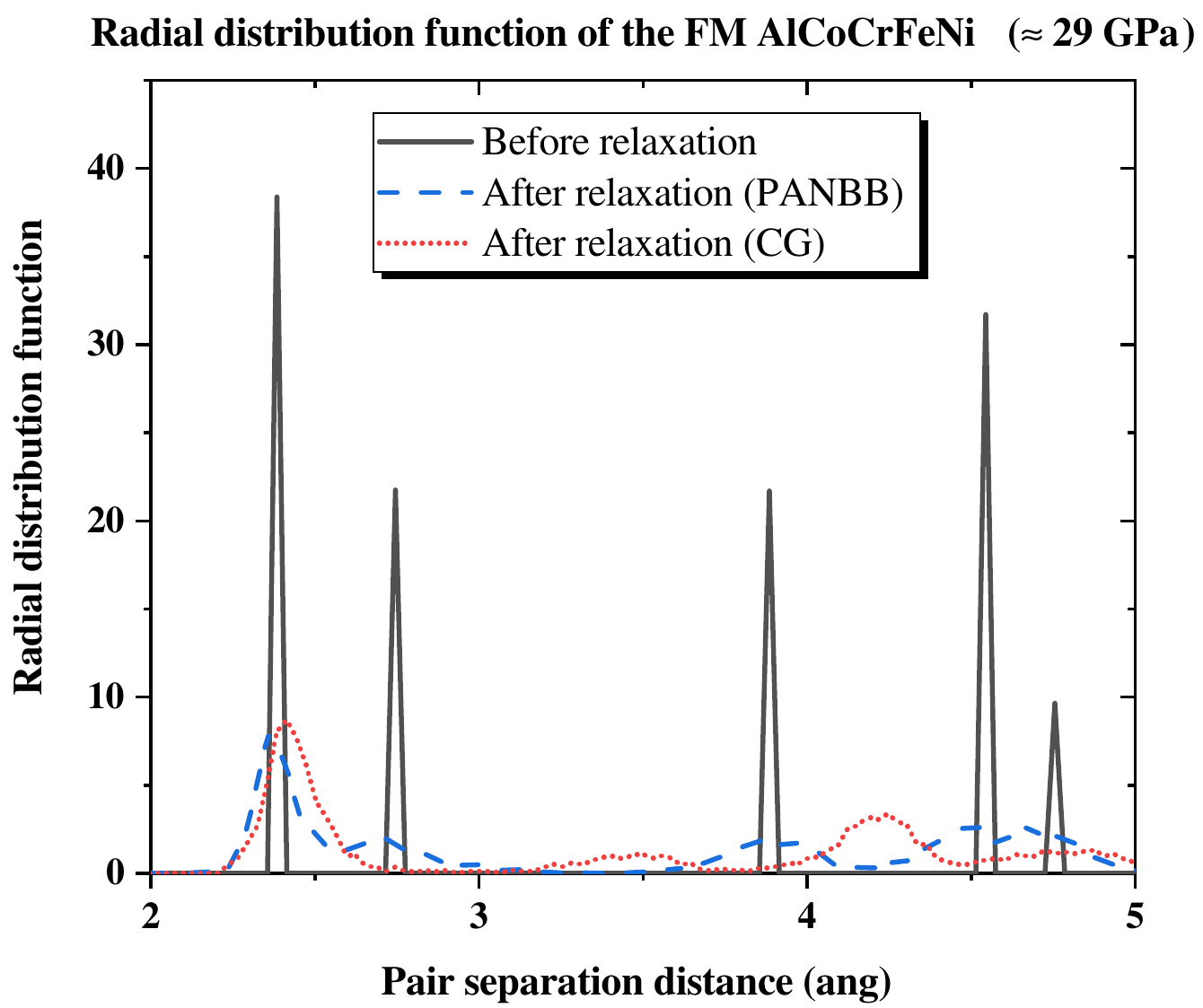}
		\caption{Radial distribution functions of the FM AlCoCrFeNi at pressures of approximately 0 GPa (left), 21 GPa (middle), and 29 GPa (right). The black solid lines refer to the functions before relaxation, while the blue dashed and red dotted lines represent the functions after relaxation using PANBB and CG, respectively.}
		\label{fig:rdf}
	\end{figure*}
	
	\section{Conclusions}
	
	\par In this work, we introduce PANBB, a novel force-based algorithm for relaxing crystal structures under a fixed unit cell volume. By incorporating the gradient projections, distinct curvature-aware initial trial step sizes, and a nonmonotone LM criterion, PANBB exhibits universal superiority over the widely used methods in both efficiency and robustness, as demonstrated across a benchmark test. To the best of our knowledge, PANBB is the first algorithm that guarantees theoretical convergence for the determinant-constrained optimization. Moreover, the application in the AlCoCrFeNi HEAs also showcases its potential in facilitating calculations of the wide-range EOS and thermodynamical properties for multicomponent alloys.
	
	\begin{acknowledgements}
		The work of the third author was supported in part by the National Natural Science Foundation of China under Grant No. U23A20537. The work of the fourth author was partially supported by the National Natural Science Foundation of China under Grants No. 12125108, No. 12226008, No. 11991021, No. 11991020, No. 12021001, and No. 12288201, Key Research Program of Frontier Sciences, Chinese Academy of Sciences under Grant No. ZDBS-LY-7022, CAS-Croucher Funding Scheme for Joint Laboratories ``CAS AMSS-PolyU Joint Laboratory of Applied Mathematics: Nonlinear Optimization Theory, Algorithms and Applications''. The computations were done on the high-performance computers of State Key Laboratory of Scientific and Engineering Computing, Chinese Academy of Sciences and the clusters at Northwestern Polytechnical University. The authors thank Zhen Yang for the helpful discussions on the MMFP approach.
	\end{acknowledgements}
	
	\appendix
	
	\section{Convergence Analyses}\label{appsec:convergence analysis}
	
	In what follows, we rigorously establish the convergence of PANBB to equilibrium states regardless of the initial configurations. The analyses depend on the following two reasonable assumptions on the potential energy functional. 
	
	\begin{assume}\label{assume:lip}
		The potential energy functional $E$ is locally Lipschitz continuously differentiable with respect to $R$ and $A$. The stress tensor $\Sigma$ is continuous with respect to $R$ and $A$. 
	\end{assume}
	
	\begin{assume}\label{assume:coercive}
		The potential energy functional $E$ is coercive over $\calF$, namely, $E(R,A)\to\infty$ as $\norm{R}_2\to\infty$ or $\norm{A}_2\to\infty$ with $\det(A)=V$. 
	\end{assume}
	
	\par For brevity, we make the following notation
	\begin{equation}
		\tilde D_{\latt}:=B^\T\tilde F_{\latt},
		\label{eqn:extra notation}
	\end{equation}
	where $\tilde F_{\latt}$ is defined in Eq. \eqref{eqn:lattice force projection}, and define the level set
	\begin{equation}
		\mathcal{L}:=\{(R,A):E(R,A)\le E^{(0)},~\det(A)=V\}.
		\label{eqn:level set}
	\end{equation}
	In view of Eqs. \eqref{eqn:block update} and \eqref{eqn:extra notation}, $\Ainterkl$ can be rewritten as
	\begin{equation}
		\Ainterkl=A^{(k)}(I+\allakl\Dlak).
		\label{eqn:projection another form}
	\end{equation}
	
	\par Assumptions \ref{assume:lip} and \ref{assume:coercive} directly imply the following lemma.
	
	\begin{lem}\label{lem:lb of min singular value}
		Suppose that Assumptions \ref{assume:lip} and \ref{assume:coercive} hold. Then there exists an $M>0$ such that if $(R,A)\in\mathcal{L}$, all of $\norm{R}_\F$, $\norm{A}_\F$, $\norm{B}_\F$, $\norm{F_{\atom}}_\F$, and $\snorm{\tilde D_{\latt}}_\F$ are upper bounded by $M$.
	\end{lem}
	\begin{proof}
		The existence of a uniform upper bound for $\norm{R}_\F$ and $\norm{A}_\F$ follows directly from Assumption \ref{assume:coercive} and the definition of the set $\mathcal{L}$ in Eq.~\eqref{eqn:level set}. That for $\norm{B}_\F$ is derived from $\det(A)=V$. The uniform boundedness of $\norm{F_{\atom}}_\F$ and $\snorm{\tilde D_{\latt}}_\F$ are then deduced from Eqs. \eqref{eqn:lattice force}, \eqref{eqn:extra notation}, and Assumption \ref{assume:lip}.
	\end{proof}
	
	\par Let $\bar\alpha_{\atom}$, $\bar\alpha_{\latt}$, $\uline{\alpha}_{\atom}$, and $\uline{\alpha}_{\atom}$ denote the upper and lower bounds for the ABB step sizes in Eq. \eqref{eqn:truncation}. Based on Assumptions \ref{assume:lip} and \ref{assume:coercive}, some positive constants are defined using $M$ in Lemma \ref{lem:lb of min singular value}, $\delta_{\atom}$, $\delta_{\latt}$ in Eq. \eqref{eqn:backtracking} as well as $\bar\alpha_{\atom}$ and $\bar\alpha_{\latt}$ in Eq. \eqref{eqn:truncation}. 
	\begin{widetext}
		\begin{equation}\begin{aligned}
				L_{\nabla E}:&=\sup\lrbrace{\frac{\snorm{\nabla E(R_1,A_1)-\nabla E(R_2,A_2)}_\F}{\snorm{(R_1,A_1)-(R_2,A_2)}_\F}:\snorm{R_i}_\F\le(1+\bar\alpha_{\atom})M,~\snorm{A_i}_\F\le\frac{3}{2}M,~i=1,2},\\
				L_{\latt}:&=\sup\lrbrace{\frac{\abs{E(R,A_1)-E(R,A_2)}}{\snorm{A_1-A_2}_\F}:\snorm{R}_\F\le(1+\bar\alpha_{\atom})M,~\snorm{A_i}_\F\le3M,~i=1,2},\\
				\bar M:&=\frac{48}{7}M^3L_{\latt},\quad\bar\ell:=\left\lceil\max\lrbrace{\frac{\log\frac{L_{\nabla E}\bar\alpha_{\atom}}{2(1-\eta)}}{\log\frac{1}{\delta_{\atom}}},\frac{\log\frac{(L_{\nabla E}+2\bar M)\bar\alpha_{\latt}}{2(1-\eta)}}{\log\frac{1}{\delta_{\latt}}},\frac{\log(2M\bar\alpha_{\latt})}{\log\frac{1}{\delta_{\latt}}}}\right\rceil.
			\end{aligned}\label{eqn:constants}\end{equation}
	\end{widetext}
	
	\par Lemma \ref{lem:successfully generation} below asserts that the LM processes in PANBB are always finite.
	
	\begin{lem}\label{lem:successfully generation}
		Suppose that Assumptions \ref{assume:lip} and \ref{assume:coercive} hold. Assume that $(R^{(k)},A^{(k)})\in\mathcal{L}$. Then in the $(k+1)$th iteration, the LM criterion \eqref{eqn:nonmonotone termination criterion} must be satisfied after invoking the backtracking \eqref{eqn:backtracking} for at most $\bar\ell$ times and $(R^{(k+1)},A^{(k+1)})\in\mathcal{L}$, where $\bar\ell$ is defined in Eq. \eqref{eqn:constants}.
	\end{lem}
	\begin{proof}
		We prove this lemma by showing that there is a sufficient reduction from $E^{(k)}$ to $E_{\trial}^{(k,\ell)}$ provided that $\ell\ge\bar\ell$, which indicates the fulfillment of the criterion \eqref{eqn:nonmonotone termination criterion}. We divide the proof into two steps. 
		
		\par\textbf{Step 1.} Estimate $E(\Rtkl,\Ainterkl)-E^{(k)}$. \\[0.2em]
		It follows from Eqs. \eqref{eqn:local linear approx} and \eqref{eqn:block update} that 
		\begin{equation}
			\begin{aligned}
				\inner{F_{\atom}^{(k)},\Rtkl-R^{(k)}}&=-\frac{1}{\alatkl}\snorm{\Rtkl-R^{(k)}}_\F^2,\\
				\inner{F_{\latt}^{(k)},\Ainterkl-A^{(k)}}&=-\frac{1}{\allakl}\snorm{\Ainterkl-A^{(k)}}_\F^2.
			\end{aligned}
			\label{eqn:optimality}
		\end{equation}
		Since $\ell\ge\bar\ell$, it is not hard to show $\allakl\le1/(2M)$. From this, $\alatkl\le\bar\alpha_{\atom}$, and Lemma \ref{lem:lb of min singular value}, we can infer $\allakl\snorm{\Dlak}_2\le1/2$ and
		\begin{align}
			\snorm{\Rtkl}_\F&\stackrel{\eqref{eqn:block update}}{\le}(1+\bar\alpha_{\atom})M,\label{eqn:R update norm}\\
			\snorm{\Ainterkl}_\F&\stackrel{\mathmakebox[\widthof{=}]{\eqref{eqn:projection another form}}}{\le}\snorm{A^{(k)}}_\F(1+\allakl\snorm{\Dlak}_2)\le\frac{3}{2}M. 
			\label{eqn:projection norm}
		\end{align}
		By the definition of $L_{\nabla E}$ in Eq. \eqref{eqn:constants}, one can derive from Eqs. \eqref{eqn:optimality}-\eqref{eqn:projection norm} that $E(\Rtkl,\Ainterkl)-E^{(k)}$ is upper bounded by
		\begin{align}
			&\lrsquare{\frac{L_{\nabla E}}{2}-\frac{1}{\alatkl}}\snorm{\Rtkl-R^{(k)}}_\F^2\nonumber\\
			&+\lrsquare{\frac{L_{\nabla E}}{2}-\frac{1}{\allakl}}\snorm{\Ainterkl-A^{(k)}}_\F^2.\label{eqn:estimate one}
		\end{align}
		
		\par \textbf{Step 2.} Estimate $E_{\trial}^{(k,\ell)}-E(\Rtkl,\Ainterkl)$. \\[0.2em]
		First notice from Eq. \eqref{eqn:projection another form} that $\det(\Ainterkl)=V\det(I+\allakl\Dlak)$. Due to $\allakl\snorm{\Dlak}_2\le1/2$ and the fact that the modulus of any eigenvalue is not larger than the maximum singular value, one has 
		\begin{equation}
			\det(\Ainterkl)\in\lrsquare{\frac{1}{8}V,\frac{27}{8}V}.
			\label{eqn:determinant range}
		\end{equation}
		By the definition of the scaling operation (see S4 in Fig. \ref{fig:existing methods flowchart}), Eqs. \eqref{eqn:projection norm} and \eqref{eqn:determinant range}, $\snorm{\Atkl}_\F\le3M$. From this, 
		\begin{align}
			&E_{\trial}^{(k,\ell)}-E(\Rtkl,\Ainterkl)\nonumber\\
			\stackrel{\mathmakebox[\widthof{=}]{\eqref{eqn:constants},\eqref{eqn:R update norm}}}{\le}&~L_{\latt}\snorm{\Atkl-\Ainterkl}_\F\nonumber\\
			=&~L_{\latt}\abs{\sqrt[3]{\det(\Ainterkl)}-\sqrt[3]{V}}\frac{\snorm{\Ainterkl}_\F}{\sqrt[3]{\det(\Ainterkl)}}.\label{eqn:estimate two pre}
		\end{align}
		For one thing, in view of Eqs. \eqref{eqn:projection norm} and \eqref{eqn:determinant range},
		\begin{equation}
			\frac{\snorm{\Ainterkl}_\F}{\sqrt[3]{\det(\Ainterkl)}}\le\frac{3M}{\sqrt[3]{V}}.
			\label{eqn:estimate 2}
		\end{equation}
		For another, again by Eq. \eqref{eqn:determinant range},
		\begin{align}
			&\abs{\sqrt[3]{\det(\Ainterkl)}-\sqrt[3]{V}}\nonumber\\
			=&\frac{\abs{\det(\Ainterkl)-V}}{\det(\Ainterkl)^{2/3}+\det(\Ainterkl)^{1/3}V^{1/3}+V^{2/3}}\nonumber\\
			\le&\frac{4\abs{\det(\Ainterkl)-V}}{7V^{2/3}}.\label{eqn:estimate 3}
		\end{align}
		Let $\lambda_1$, $\lambda_2$, $\lambda_3$ be the eigenvalues of $\allakl\Dlak$. Since $\aaFlak$ lies on $\calT^{(k)}$ (see the definition in Eq. \eqref{eqn:local linear approx}), one has $\sum_{i=1}^3\lambda_i=0$ from Eq. \eqref{eqn:extra notation}. Therefore,
		\begin{align}
			\abs{\det(\Ainterkl)-V}&\stackrel{\mathmakebox[\widthof{=}]{\eqref{eqn:projection another form}}}{=\:\:}V\abs{\det(I+\allakl\Dlak)-1}\nonumber\\
			&=\:\:V\abs{\lambda_1\lambda_2+\lambda_1\lambda_3+\lambda_2\lambda_3+\lambda_1\lambda_2\lambda_3}.\nonumber
		\end{align}
		Since $\max_i\{\abs{\lambda_i}\}\le\allakl\snorm{\Dlak}_2\le1/2$, 
		\begin{align}
			\abs{\det(\Ainterkl)-V}&\le4V\max_i\{\abs{\lambda_i}\}^2\le4V\norm{\allakl\Dlak}_2^2\nonumber\\
			&\stackrel{\mathmakebox[\widthof{=}]{\eqref{eqn:projection another form}}}{=}4V\norm{B^{(k)\T}(\Ainterkl-A^{(k)})}_2^2\nonumber\\
			&\le4M^2V\snorm{\Ainterkl-A^{(k)}}_\F^2,\label{eqn:estimate 4}
		\end{align}
		where the last inequality uses Lemma \ref{lem:lb of min singular value}. Combining Eq. \eqref{eqn:constants} and Eqs. \eqref{eqn:estimate two pre}-\eqref{eqn:estimate 4}, one obtains
		\begin{equation}
			E_{\trial}^{(k,\ell)}-E(\Rtkl,\Ainterkl)\le\bar M\snorm{\Ainterkl-A^{(k)}}_\F^2.
			\label{eqn:estimate two}
		\end{equation}
		
		\par Finally, putting Eqs. \eqref{eqn:estimate one} and \eqref{eqn:estimate two} together, we have
		\begin{multline*}
			E_{\trial}^{(k,\ell)}-E^{(k)}
			\le\lrsquare{\frac{L_{\nabla E}}{2}-\frac{1}{\alatkl}}\snorm{\Rtkl-R^{(k)}}_\F^2\\
			+\lrsquare{\frac{L_{\nabla E}}{2}+\bar M-\frac{1}{\allakl}}\snorm{\Ainterkl-A^{(k)}}_\F^2.
		\end{multline*}
		
		\par By virtue of the criterion \eqref{eqn:nonmonotone termination criterion} and $E^{(k)}\le\bar E^{(k)}$ \cite[Lemma 2]{hu2022force}, the LM process must stop when
		$$\left\{\begin{array}{l}
			\dfrac{L_{\nabla E}}{2}-\dfrac{1}{\alatkl}\le-\dfrac{\eta}{\alatkl},\\
			\dfrac{L_{\nabla E}}{2}+\bar M-\dfrac{1}{\allakl}\le-\dfrac{\eta}{\allakl},
		\end{array}\right.$$
		which is fulfilled once $\ell$ reaches $\bar\ell$. Since $E(\Rtkl,\Atkl)\le\bar E^{(k)}\le E^{(0)}$, $(R^{(k+1)},A^{(k+1)})\in\mathcal{L}$. 
	\end{proof}
	
	\par Leveraging Lemma \ref{lem:successfully generation}, we are ready to establish the convergence of PANBB.
	
	\begin{thm}\label{thm:convergence}
		Suppose that Assumptions \ref{assume:lip} and \ref{assume:coercive} hold. Let $\{(R^{(k)},A^{(k)})\}$ be the configuration sequence generated by \textup{PANBB} with $\mu_k\in[\mu_{\min},\mu_{\max}]\subseteq(0,1]$. Then we have $F_{\atom}^{(k)}\to0$ and $\aaFlak\to0$ as $k\to\infty$. Moreover, there exists at least one limit point of $\{(R^{(k)},A^{(k)})\}$ and any limit point $(R^{\star},A^{\star})$ satisfies $F_{\atom}^{\star}=0$, $\Sigma_{\dev}^{\star}=0$. 
	\end{thm}
	\begin{proof}
		By Lemma \ref{lem:successfully generation}, the backtracking scheme \eqref{eqn:backtracking}, and $\alpha_{\atom}^{(k-1,0)}\ge\uline{\alpha}_{\atom}$, $\alpha_{\latt}^{(k-1,0)}\ge\uline{\alpha}_{\latt}$, it holds that 
		$$\alpha_{\atom}^{(k)}\ge\beta_{\atom}:=\uline{\alpha}_{\atom}\delta_{\atom}^{\bar\ell},~~\alpha_{\latt}^{(k)}\ge\beta_{\latt}:=\uline{\alpha}_{\latt}\delta_{\latt}^{\bar\ell}.$$
		Plugging this into the criterion \eqref{eqn:nonmonotone termination criterion}, we obtain
		\begin{equation}
			E^{(k+1)}\le\bar E^{(k)}-\eta\lrsquare{\beta_{\atom}\snorm{F_{\atom}^{(k)}}_\F^2+\beta_{\latt}\snorm{\aaFlak}_\F^2}.\label{eqn:suff reduce nonmonotone}
		\end{equation}
		Let $\Delta^{(k)}:=\beta_{\atom}\snorm{F_{\atom}^{(k)}}_\F^2+\beta_{\latt}\snorm{\aaFlak}_\F^2$. From this, a recursion for the monitoring sequence follows.
		\begin{align}
			\bar E^{(k+1)}&\stackrel{\mathmakebox[\widthof{=}]{\eqref{eqn:recursion}}}{=\:\:}\frac{\bar E^{(k)}+\mu^{(k)}q^{(k)}E^{(k+1)}}{q^{(k+1)}}\nonumber\\
			&\stackrel{\mathmakebox[\widthof{=}]{\eqref{eqn:suff reduce nonmonotone}}}{\le\:\:}\frac{\bar E^{(k)}+\mu^{(k)}q^{(k)}(\bar E^{(k)}-\eta\Delta^{(k)})}{q^{(k+1)}}\nonumber\\
			&=\:\:\bar E^{(k)}-\eta\frac{\mu^{(k)}q^{(k)}}{q^{(k+1)}}\Delta^{(k)}.\label{eqn:recursion of monitoring}
		\end{align}
		A byproduct of Eq. \eqref{eqn:recursion of monitoring} is the monotonicity of $\{\bar E^{(k)}\}$. Since $E$ is bounded from below over $\mathcal{L}$ (by Assumption \ref{assume:lip} and Lemma \ref{lem:lb of min singular value}) and $E^{(k)}\le\bar E^{(k)}$ for any $k\ge0$ \cite[Lemma 2]{hu2022force}, $\{\bar E^{(k)}\}$ is also bounded from below. Consequently,
		\begin{equation}
			\sum_{k=0}^\infty\frac{\mu^{(k)}q^{(k)}\Delta^{(k)}}{q^{(k+1)}}\stackrel{\eqref{eqn:recursion of monitoring}}{\le}\frac{1}{\eta}\sum_{k=0}^\infty(\bar E^{(k)}-\bar E^{(k+1)})<+\infty.
			\label{eqn:res summation}
		\end{equation}
		Following from Eq. \eqref{eqn:recursion} and $\mu_{\max}\ge\mu^{(k)}\ge\mu_{\min}>0$, $q^{(k)}\ge1$, we can derive
		\begin{equation*}
			\frac{\mu^{(k)}q^{(k)}}{q^{(k+1)}}\ge\frac{\mu_{\min}q^{(k)}}{\mu^{(k)}q^{(k)}+1}=\frac{\mu_{\min}}{\mu^{(k)}+1/q^{(k)}}\ge\frac{\mu_{\min}}{\mu_{\max}+1},
		\end{equation*}
		which, together with Eq. \eqref{eqn:res summation}, yields $\Delta^{(k)}\to0$ as $k\to\infty$. By $\beta_{\atom}$, $\beta_{\latt}>0$, the definition of $\Delta^{(k)}$, and Lemma \ref{lem:lb of min singular value}, one has $F_{\atom}^{(k)}\to0$ and $\aaFlak\to0$ as $k\to\infty$. 
		
		\par Let $(R^{\star},A^{\star})$ be a limit point of $\{(R^{(k)},A^{(k)})\}$, whose existence is guaranteed by Lemma \ref{lem:lb of min singular value}. Then by the above arguments and Eq. \eqref{eqn:lattice force projection}, $F_{\atom}^{\star}=0$ and
		\begin{align*}
			0=\tilde F_{\latt}^{\star}&=F_{\latt}^{\star}-\frac{\inner{B^{\star},F_{\latt}^{\star}}}{\snorm{B^{\star}}_\F^2}B^{\star}\\
			&=V\lrsquare{\Sigma^{\star}B^{\star}-\frac{\inner{B^{\star},\Sigma^{\star}B^{\star}}}{\snorm{B^{\star}}_\F^2}B^{\star}},
		\end{align*}
		which implies that $\Sigma^{\star}$ is a multiple of $I$ and hence $\Sigma_{\dev}^{\star}=0$. 
	\end{proof}
	
	\section{Magnetic Disorder Modeling}\label{appsec:modeling details}
	
	We elaborate below on the technical details when modeling the magnetic disorder for the PM structures, in particular, how to solve the distribution problem for the spin-up or spin-down sites. We assume that an atomic configuration is available (e.g., given by the SAE method), and moreover introduce some notations:
	
	\begin{itemize}
		\item $S\in\N$: number of the atomic species (in the AlCoCrFeNi case, equals 5);
		
		\item $S_{\magg}\in\N$: number of the magnetic atomic species (in the AlCoCrFeNi case, equals 4). For convenience, we assume that species $1,\ldots,S_{\magg}$ are magnetic;
		
		\item $N_s\in\N$: number of the atoms in the unit cell belonging to species $s$ (in the AlCoCrFeNi case, $N_1=\cdots=N_5=32$) ($s\in\{1,\ldots,S\}$);
		
		\item $N_{\magg}\in\N$: number of magnetic atoms in the unit cell (in the AlCoCrFeNi case, equals 64); 
		
		\item $\ell_{\ax}\in\N$: number of the atomic layers in the unit cell with axis $\ax\in\{a,b,c\}$ as normal;
		
		\item $\calI_s\subseteq\{1,\ldots,N\}$: indices of the sites occupied by species $s$ ($s\in\{1,\ldots,S_{\magg}\}$);
		
		\item $\calI_{\ax,l,s}\subseteq\{1,\ldots,N\}$: indices of the sites occupied by species $s$ located on the $l$th atomic layer with axis ax as normal ($\ax\in\{a,b,c\}$, $l\in\{1,\ldots,\ell_{\ax}\}$, $s\in\{1,\ldots,S_{\magg}\}$);
		
		\item $\y\in\{0,1\}^{N_{\magg}}$: binary indicators for the spin-up configuration, where $y_n=1$ ($n\in\{1,\ldots,N_{\magg}\}$) means that the $n$th atom spins up. Conversely, one may interpret the zero entries as spin-down atoms;
		
		\item $\y^0\in\{0,1\}^{N_{\magg}}$: binary indicators for the trial spin-up configuration. The meaning of each entry is analogous to that of $\y$. 
	\end{itemize}
	
	\par With the above notations, we can formulate the spin-up distribution problem as the following binary integer programming:
	\begin{equation}
		\begin{array}{cl}
			\max_{\y} & \inner{\y^0,\y}, \\
			\st & \lfloor N_s/(2\ell_{\ax})\rfloor\le\sum_{n\in\calI_{\ax,l,s}}y_n\le\lceil N_s/(2\ell_{\ax})\rceil,\\
			& l\in\{1,\ldots,\ell_{\ax}\},~\ax\in\{a,b,c\},~s\in\{1,\ldots,S_{\magg}\},\\
			& \sum_{n\in\calI_s}y_n=N_s/2,~s\in\{1,\ldots,S_{\magg}\},\\
			& \y\in\{0,1\}^{N_{\magg}}.
		\end{array}
		\label{eqn:ip spinup}
	\end{equation}
	
	\par The integer programming \eqref{eqn:ip spinup} maximizes the similarity between the trial and optimized spin-up structures, represented by the inner product of their corresponding binary indicators. In our experiments, we simply take $\y^0=0$. The first line of constraints enforces that the spin-up sites of each species are evenly distributed over each atomic layer. The second line of constraints imposes that the numbers of spin-up and spin-down sites for each species are identical. Upon solving the integer programming \eqref{eqn:ip spinup} (e.g., by off-the-shelf software like \textsc{matlab} \cite{MATLAB}), one can transform the solution $\y^{\star}$ to a concrete magnetic configuration. 
	
	\normalem

\begin{thebibliography}{52}
\bibitem [{Notea)()}]{Notea)}%
		\BibitemOpen
		\bibinfo {note} {The specific initial step size for the atomic part is also
			used by CG realized in \protect \textsc {vasp} and \protect \textsc {cessp}.
			The numerical results of CG and PANBB will become more comparable. A small
			initial step size is chosen for the lattice part to prevent unexpected
			breakdown. In simulations, the performance of PANBB are not sensitive to the
			two values.}\BibitemShut {Stop}%
\bibitem [{Noteb)()}]{Noteb)}%
		\BibitemOpen
		\bibinfo {note} {We use a small $\delta _{\protect \textup {\protect \text
					{atom}}}$ to end the LM process more quickly and help the ensuing ABB step
			sizes learn local curvatures better. The value of $\delta _{\protect \textup
				{\protect \text {latt}}}$ is not so critical, since usually $\alpha
			_{\protect \textup {\protect \text {latt}}}^{(k-1,0)}$ is already
			sufficiently small.}\BibitemShut {Stop}%
\end{thebibliography}
	
	%

\end{document}